\documentclass[cmp]{svjour}

\usepackage{amsmath}
\usepackage{amssymb}
\usepackage{graphicx}

\usepackage{epsfig}
\usepackage{algpseudocode}
\usepackage{amscd}
\usepackage{algorithm}
\usepackage{tikz}
\usepackage{enumerate}

\newcommand{\beq}{\begin{equation}}
\newcommand{\eeq}{\end{equation}}

\newcommand{\beqa}{\begin{eqnarray}}
\newcommand{\eeqa}{\end{eqnarray}}

\newcommand{\bal}{\begin{align}}
\newcommand{\eal}{\end{align}}

\newcommand{\bsp}{\begin{equation}\begin{split}}
\newcommand{\esp}{\end{split}\end{equation}}

\newcommand{\bit}{\begin{enumerate}[(i)]}
\newcommand{\eit}{\end{enumerate}}

\newcommand{\ben}{\begin{enumerate}[(i)]}
\newcommand{\een}{\end{enumerate}}

\newcommand{\nn}{\nonumber}

\newcommand{\SP}[2]{\langle #1,#2 \rangle}
\newcommand{\AR}{\mathbb{R}}

\newcommand{\HR}{\mathcal{H}}

\newcommand{\id}{\mathbb{I} }

\newcommand{\tr}{\mathrm{tr}}
\newcommand{\rank}{\mathrm{rank}}

\newcommand{\Herm}{\mathrm{Herm}}
\newcommand{\data}{\mathcal{D}}
\newcommand{\Pst}{P_{\mathrm{st}}}
\newcommand{\Pm}{P_{\mathrm{m}}}
\newcommand{\curlyPst}{\mathcal{P}_{\mathrm{st}}}
\newcommand{\curlyPm}{\mathcal{P}_{\mathrm{m}}}

\newcommand{\Dsquare}{\mathcal{D}_{\square}}

\newcommand{\thetasharp}{\theta^{(\sharp)}}
\newcommand{\knowledge}{\vec{\mathcal{K}}}
\newcommand{\Sym}{\mathrm{Sym}}

\def\mc#1{\multicolumn{1}{c|}{#1}}
\def\cm#1{\multicolumn{1}{|c}{#1}}

\algnewcommand\algorithmicswitch{\textbf{switch}}
\algnewcommand\algorithmiccase{\textbf{case}}
\algnewcommand\algorithmicassert{\texttt{assert}}
\algnewcommand\Assert[1]{\State \algorithmicassert(#1)}%
\algdef{SE}[SWITCH]{Switch}{EndSwitch}[1]{\algorithmicswitch\ #1\ \algorithmicdo}{\algorithmicend\ \algorithmicswitch}%
\algdef{SE}[CASE]{Case}{EndCase}[1]{\algorithmiccase\ #1}{\algorithmicend\ \algorithmiccase}%
\algtext*{EndSwitch}%
\algtext*{EndCase}%
\begin{document}
\title{Global Completability with Applications to Self-Consistent Quantum Tomography}
\author{Cyril Stark
}                     
\institute{Theoretische Physik, ETH Zurich, CH-8093 Zurich, Switzerland; E-mail: starkc@phys.ethz.ch}
\date{Received: date / Accepted: date}
%
\communicated{name}
\maketitle
\begin{abstract}

Let $\vec{p}_{1}, ..., \vec{p}_{N} \in \AR^D$ be unknown vectors and let $\Omega \subseteq \{1,...,N\}^{\times 2}$. Assume that the inner products $\vec{p}_{i}^T \vec{p}_{j}$ are fixed for all $(i,j) \in \Omega$. Do these inner product constraints (up to simultaneous rotation of all vectors) determine $\vec{p}_{1}, ..., \vec{p}_{N}$ uniquely? Here we derive a necessary and sufficient condition for the uniqueness of $\vec{p}_{1}, ...,\vec{p}_{N}$ (i.e., global completability) which is applicable to a large class of practically relevant sets $\Omega$. Moreover, given $\Omega$, we show that the condition for global completability is universal in the sense that for almost all vectors $\vec{p}_{1}, ...,\vec{p}_{N} \in \AR^D$ the completability of $\vec{p}_{1}, ...,\vec{p}_{N}$ only depends on $\Omega$ and not on the specific values of $\vec{p}_{i}^T \vec{p}_{j}$ for $(i,j) \in \Omega$. This work was motivated by practical considerations, namely, self-consistent quantum tomography.

\end{abstract}

%

\section{Introduction}\label{Rigidity:generic.configurations}

Assume we have built an experiment that allows us to prepare $W$ states $(\rho_{w})_{w=1}^W$ and to perform $V$ measurements $\bigl( (E_{vk})_{k=1}^K \bigr)_{v=1}^V$. Here, ``$v$'' enumerates the different measurements and ``$k$'' enumerates the outcomes of each measurements. To simplify the notation we assume that each measurement has the same number of outcomes $K$. Throughout we assume that the underlying Hilbert space is finite-dimensional and therefore, all states and each of the measurement operators are elements of the space of Hermitian matrices $\Herm(\mathbb{C}^d)$ on a $d$-dimensional Hilbert space $\mathbb{C}^d$. We equip the real $d^2$-dimensional vector space $(\Herm(\mathbb{C}^d), \SP{\cdot}{\cdot})$ with the Hilbert-Schmidt inner product $\SP{A}{B} = \tr(A^* B)$ and choose an arbitrary orthonormal basis in $\Herm(\mathbb{C}^d)$. With respect to this basis we can represent each state and each measurement operator as column vectors $(\vec{\rho}_{w})_{w=1}^W$ and $\bigl( (\vec{E}_{vk})_{k=1}^K \bigr)_{v=1}^V$ in $\AR^{d^2}$. In the remainder, $D := d^2$ and we assume that the states $(\rho_{w})_{w}$ and the measurements $(E_{vk})_{vk}$ linearly span $\Herm(\mathbb{C}^d)$. Due to the orthonormality of our reference basis,
\beq\label{temp.82394002A}
	\tr( \rho_{w} E_{vk} ) = (\vec{\rho}_{w})^T \vec{E}_{vk}.
\eeq
Define
\beq\label{definition.of.n.in.the.flexibility.context}
	P = ( P_{\mathrm{st}} \, |  \, P_{\mathrm{m}}  )
\eeq
where 
\[
	P_{\mathrm{st}} = ( \vec{\rho}_{1} \, | \cdots | \, \vec{\rho}_{W}), \ P_{\mathrm{m}} = (\vec{E}_{11}  \, | \cdots | \, \vec{E}_{1K} \, | \cdots  | \, \vec{E}_{V1} \, | \cdots | \, \vec{E}_{VK})
\]
specify all of the prepared states, and all of the performed measurements, respectively. The Gram matrix $G \in \AR^{N \times N}$, $N = W + VK$, of all of the states and measurement operators thus satisfies
\beq\label{temp.82394002B}
	G = P^{T}P = \left( \begin{array}{cc}  G_{\mathrm{st}} & \mathcal{D} \\   \mathcal{D}^T & G_{\mathrm{m}}   \end{array} \right)
\eeq
where $G_{\mathrm{st}} = P_{\mathrm{st}}^T P_{\mathrm{st}}$ is the Gram matrix of all the states, and $G_{\mathrm{m}} = P_{\mathrm{m}}^T P_{\mathrm{m}}$ is the Gram matrix of all the measurement operators. The off-diagonal block $\data$ captures the phenomenological behavior of the experiment. More precisely, by Born's rule, $\data_{w,n_{vk}}$ ($n_{vk} = (v-1)K + k$) is the probability for measuring outcome ``$k$'' given that we have prepared the state ``$w$'' and performed the measurement ``$v$''. It follows that $\data$ can (in principle) be determined experimentally if each measurement can be repeated infinitely many times. By the assumption that the states and the measurements linearly span $\Herm(\mathbb{C}^d)$,
\beq\label{Eq:the.spanning.assumption}
	\rank(G) = \rank(P) = D.
\eeq

\emph{Generic states and measurements.} Now suppose we would have intended to construct our experiment such that it allows the preparation of states and the performance of measurements described by $P_{\mathrm{theory}} \in \AR^{D \times N}$. We will never succeed precisely; the best quantum mechanically valid approximation $P$ in $\AR^{D \times N}$ of what is actually happening in the experiment is a random quantity drawn from a probability measure $\mu$ on $\AR^{D \times N}$. Let $\lambda$ denote the Lebesgue measure on $\AR^{D \times N}$, and define
\beq\label{Rigidity:def.of.set.L}
	\mathcal{L} := \bigl\{ \mu \text{ is a measure on $\AR^{D \times N}$} \; \bigl| \; \mu \ll \lambda, \; \mu\bigl(\AR^{D \times N}\bigr) = 1 \bigr\}
\eeq
i.e., $\mathcal{L}$ denotes the set of probability measures on $\AR^{D \times N}$ which are absolutely continuous with respect to the Lebesgue measure. In the remainder, experiments will be called \emph{generic} if\footnote{Hence, generically, density matrices are full-rank and measurements are described in terms of full-rank elements of a positive operator valued measure (POVM, see~\cite{NielsenChuang2000}).}
\beq\label{Rigidity:the.abs.cont.assumption}
	\mu \in \mathcal{L}.
\eeq

Is the measurement data $\data$ sufficient to determine uniquely its theoretical description in terms states $(\rho_{w})_{w}$ and measurements $(E_{vk})_{vk}$? Obviously, the answer is no. Let $U \in \mathbb{C}^{d \times d}$ be unitary, and let $(\rho_{w})_{w=1}^W$, $\bigl( (E_{vk})_{k=1}^K \bigr)_{v=1}^V$ be a valid explanation of the measurement data~$\data$. Then, $(U\rho_{w}U^*)_{w}$, $(UE_{vk}U^*)_{vk}$ is an equally valid quantum model for~$\data$. Hence, the states and measurements are never uniquely determined by the measured data. However, the linear transformations $\Herm(\mathbb{C}^d) \rightarrow \Herm(\mathbb{C}^d)$, $A \mapsto UAU^*$ ($U$ unitary) are special instances of orthogonal transformations in $(\Herm(\mathbb{C}^d, \sp{\cdot}{\cdot}))$. The Gram matrix $G$ associated to the states and the measurements is invariant under the simultaneous rotation of all the states and measurements. On the other hand, the Gram matrix specifies the states and measurements uniquely up to the simultaneous rotation of all the states and measurements. In this sense, $G$ determines uniquely all the pairwise geometric relationships between states and states, between states and measurements, and between measurements and measurements. 

\emph{Uniqueness of $G$?} It is natural to ask whether measurement data $\data$ suffices to uniquely determine the state-measurement Gram matrix $G$. A simple construction shows (see~\cite{Stark2012GramEstimations}) that the answer is no independently of $W$, $V$, $K$, and the Hilbert space dimension $d$ if all the states and measurement operators are full-rank matrices. For generic experiments, all of the states and measurements can be described by full-rank matrices in $\Herm(\mathbb{C}^d)$. However, in idealized situations involving rank-deficient states and measurements, the measured data~$\data$ can suffice to determine $G$ uniquely (see~\cite{Stark2012RigidityConsiderations}). Here, we are interested in generic experiments, and if the data~$\data$ is insufficient to determine $G$, we should ask: what are the necessary and sufficient conditions for the uniqueness of $G$?

In this paper we are regarding the states and the measurements as unconstrained vectors $(\vec{\rho}_{w})_{w=1}^W$ and $\bigl( (\vec{E}_{vk})_{k=1}^K \bigr)_{v=1}^V$ in $\AR^{D}$, i.e., we are disregarding the quantum mechanical constraints $\rho_{w} \geq 0$, $\tr(\rho_{w}) = 1$, $E_{vk} \geq 0$ and $\sum_{k} E_{vk} = \id$. Obviously, uniqueness of $G$ when disregarding the quantum mechanical constraints implies uniqueness of $G$ when taking into account the quantum mechanical constraints. Therefore, the criteria for the uniqueness of $G$ that we present in this paper are sufficient but not necessary from the quantum mechanical perspective but they are necessary and sufficient when regarding $(\vec{\rho}_{w})_{w=1}^W$ and $\bigl( (\vec{E}_{vk})_{k=1}^K \bigr)_{v=1}^V$ as unconstrained vectors in $\AR^{D}$ (relevant for applications outside of physics).

We have arrived at the following geometric problem: Let $G$ with $\rank(G) = D$ denote the Gram matrix associated to the column vectors of \emph{any} matrix $P \in \AR^{D \times N}$ (cf.~\eqref{definition.of.n.in.the.flexibility.context}), and let $\Omega \subseteq \{ 1,...,N \}^{\times 2}$ denote a subset of its entries. For which $\Omega$ and linear constraints $G_{\Omega} = \knowledge$ ($\knowledge \in \AR^{| \Omega |}$) is $G$ determined uniquely?  

In accordance with~\cite{SingerCucuringu} we call configurations $P \in \AR^{D \times N}$ satisfying $G_{\Omega} = (P^TP)_{\Omega} = \knowledge$ \emph{locally completable} if there exists an open neighborhood $U$ of $P$ such that up to trivial transformations,\footnote{Trivial transformations are $P \mapsto OP$ for $O$ orthogonal on $\AR^D$.} $P$ is the only configuration satisfying the inner product constraints $G_{\Omega} = \knowledge$. In contrast, the configuration $P \in \AR^{D \times N}$ satisfying $G_{\Omega} = (P^TP)_{\Omega} = \knowledge$ is \emph{globally completable} if (up to trivial transformations) $P$ is uniquely determined by the inner product constraints and therefore, $G$ is uniquely determined.

\emph{Conventional rigidity theory.} A closely related question is raised in conventional rigidity theory. Let $\vec{p}_{1}, ...,\vec{p}_{N} \in \AR^D$ be some unknown points. Let $S \in \AR^{N \times N}$ be defined by $S_{ij} = \| \vec{p}_{i} - \vec{p}_{j} \|$, and let $\Omega \subseteq \{ 1,...,N \}^{\times 2}$. Given $S_{\Omega}$, is $S$ determined uniquely, i.e., is $\vec{p}_{1}, ...,\vec{p}_{N} \in \AR^D$ uniquely determined up to simultaneous rotations and translations of all the points $\vec{p}_{1}, ...,\vec{p}_{N}$? This question has a long history. First discussions of the rigidity of $\vec{p}_{1}, ...,\vec{p}_{N}$ given $S_{\Omega}$ date back to Euler and Cauchy (see~\cite{euler1862opera,cauchy1813polygones}). A necessary and sufficient conditions for the exclusion of smooth deformations of generic points $\vec{p}_{1}, ...,\vec{p}_{N}$ given $S_{\Omega}$---the so called \emph{Asimow-Roth Theorem}---has been derived by Asimow and Roth (see~\cite{asimow1978rigidity} but also~\cite{Gluck1975}). The impossibility to smoothly deform $\vec{p}_{1}, ...,\vec{p}_{N}$ given $S_{\Omega}$ can be regarded as a local property at $\vec{p}_{1}, ...,\vec{p}_{N}$. One says that $\vec{p}_{1}, ...,\vec{p}_{N}$ given $S_{\Omega}$ is \emph{locally rigid}. On the other hand, proving that $\vec{p}_{1}, ...,\vec{p}_{N}$ given $S_{\Omega}$ is unique up to rotations and translations is a global property. Thus, the point cloud $\vec{p}_{1}, ...,\vec{p}_{N}$ given $S_{\Omega}$  is said to be \emph{globally rigid} if it is unique up to rigid transformations. Proving global rigidity is much more difficult than proving local rigidity. In 2005, R.~Connelly managed to derive a criterion which is sufficient to guarantee global rigidity for generic points $\vec{p}_{1}, ...,\vec{p}_{N}$ (see~\cite{connelly2005generic}). Only recently, in 2010, S.~J.~Gortler and co-workers have managed to prove that for generic points, R.~Connelly's criterion is also necessary (see~\cite{gortler2010characterizing}).

\emph{Relation to other work.} Despite its fundamental character, the analysis of the completability of vectors was motivated by practical considerations (see~\cite{Stark2012GramEstimations}), namely, self-consitent tomography, i.e., the task to fit quantum states and measurements to measured data (see~\cite{Stark2012GramEstimations,merkel2012self,James2012NJP,Mogilevtsev2010PRA,MogilevtsevHradil2009,MogilevtsevHradil2012,kimmel2013robust,medford2013self,RossetGisin2012}). Hence, self-consistent tomography is a generalization of state or measurement tomography (see for instance~\cite{BlumeKohut2010,Hradil1997,Sacchi1999,Audenaert2009,Gross2010,FlammiaGross2012}). Moreover, we would like to point out the close relation to self-testing of quantum devices (see~\cite{mayers2003self,bardyn2009device,pironio2009device,magniez2006self,mckague2012robust}) even though self-testing studies 2-party settings. Outside of physics, the uniqueness of $G$ is of fundamental interest in general data analysis (see~\cite{SingerCucuringu} and references therein). 

The seminal paper~\cite{SingerCucuringu} by A.~Singer and M.~Cucuringu was a crucial starting point for our work. For instance, Algorithm~\ref{flexibility.test} was proposed in~\cite{SingerCucuringu} to test local completability for inner product-constrained vectors. This test was shown to be sufficient for local completability. Here, we strengthen the power of their Algorithm by proving that the underlying test for local completability is for generic vectors not only sufficient but also necessary for local completability.\footnote{The proof of this is missing in~\cite{SingerCucuringu} because the provided arguments (see remarks below Eq.~(3.6) in~\cite{SingerCucuringu}) rely on the unproven inner product version of the so called Asimow-Roth Theorem.} The analysis of global completability of Gram matrices presented in~\cite{SingerCucuringu} relies on plausible conjectures (see Section~3.2 in~\cite{SingerCucuringu}). Here, we derive a necessary and sufficient criterion for global completability which is applicable for a large class of sets $\Omega$.

We would like to point out important findings by M.~Laurent and co-workers which are independent from our work. In~\cite{laurent2013positive} they analyzed the uniqueness of Gram matrices in the so called \emph{spherical} setting where the diagonal of $G$ is assumed to be known. Then, fixing the value of $\vec{p}_{i}^T \vec{p}_{j}$ determines the angle between the vectors $\vec{p}_{i}$ and $\vec{p}_{j}$. Hence, regarding angles as distances between vectors of unit lengths, the analysis of the spherical setting is closely related to conventional rigidity theory.

\emph{Our contributions and organization of the paper.} In Section~\ref{Sec:sufficient.conditions} we prove that for almost all $\vec{p}_{1}, ..., \vec{p}_{N}$ the sufficient criterion for local completability which has been derived in~\cite{SingerCucuringu} is also necessary. We show that the condition is universal, i.e., almost all state-measurement configurations $P$ (recall Eq.~\eqref{definition.of.n.in.the.flexibility.context}) are either locally completable or they are not locally completable. This allows for the computation of completability phase diagrams which hold true with probability 1 for generic experiments. In Section~\ref{Section:Global.rigidity}, we present a necessary and sufficient condition for global completability of inner product-constrained vectors $\vec{p}_{1}, ..., \vec{p}_{N}$. This condition is applicable to a large class of practically relevant scenarios but---in contrast to the criterion for local completability---it is not applicable to all possible choices for the index set $\Omega$ marking the known entries of $G$. Again we show that the condition for global completability is universal, i.e., almost all state-measurement configurations $P$ (recall Eq.~\eqref{definition.of.n.in.the.flexibility.context}) are either locally completable or they are not locally completable. In Section~\ref{Sec:Conclusions}, we summarize our findings and conclude the paper.

\section{Local completability}\label{Sec:sufficient.conditions}

Recall that $\Omega \subseteq \{ 1,...,N \}^{\times 2}$ marks the set of entries of $G$ that are fixed, i.e., $G_{\Omega} = \knowledge$. Disregarding the quantum mechanical constraints $\rho_{w} \geq 0$, $\tr(\rho_{w}) = 1$, $E_{vk} \geq 0$ and $\sum_{k} E_{vk} = \id$, we can describe the states and measurements as the columns of an arbitrary matrix $P \in \AR^{D \times N}$ (recall Eq.~\eqref{definition.of.n.in.the.flexibility.context}) satisfying $D = \rank(P) = \rank(G)$ (recall Eq.~\eqref{Eq:the.spanning.assumption}). Note that postulating $D = \rank(G)$ we implicitly assume $N \geq D$. The following Theorem is crucial for the remainder of this section.

\begin{theorem}\label{Theorem.local.rigidity}
	Let $P \in \AR^{D \times N}$ be generic and $G_{\Omega}(P) = \vec{\mathcal{K}} \in \AR^{| \Omega |}$. Set $r := \rank( d(G_{\Omega})_{P} )$. Then (with probability 1) there exists an open and full measure set $U_{P}$ containing $P$ such that $U_{P} \cap G_{\Omega}^{-1}(\vec{\mathcal{K}})$ is a smooth $(DN - r)$-dimensional submanifold of $\AR^{D \times N}$.
\end{theorem}

In the Theorem, $d(G_{\Omega})_{P}$ denotes the Jacobian of the map $(\cdot)_{\Omega}: P \mapsto (P^TP)_{\Omega} \in \AR^{| \Omega |}$. We proved Theorem~\ref{Theorem.local.rigidity} because we aimed at justifying Algorithm~\ref{flexibility.test} below. However, after finishing the proof of Theorem~\ref{Theorem.local.rigidity} we noticed that our strategy to prove Theorem~\ref{Theorem.local.rigidity} is similar to the proof of the Asimow-Roth Theorem from conventional rigidity theory (see~\cite{asimow1978rigidity} but also~\cite{Gluck1975}). Therefore, we decided to move our proof of Theorem~\ref{Theorem.local.rigidity} to the appendix.

Let $O(D)$ denote the set of orthogonal matrices in $\AR^{D \times D}$. Given $P \in \AR^{D \times N}$ with $\rank(P) = D$, let 
\beq\label{Rigidity:Definition.of.trivial.mfd.N.tilde}
	\tilde{N}_{P} := O(D) P  =  \{ OP | O \in O(D)  \}
\eeq
denote the set of state-measurement configurations that can be reached by trivial transformations, i.e., simultaneous rotation of all columns of $P$. By Lemma~\ref{Appendix:Lemma:N.tilde.is.submanifold} in the appendix, $\tilde{N}_{P}$ is a $\frac{1}{2}D(D-1)$-dimensional submanifold of $\AR^{DÊ\times N}$. The comparison of the manifold $U_{P} \cap G_{\Omega}^{-1}(\vec{\mathcal{K}})$ with the trivial manifold $\tilde{N}_{P}$ will lead to a criterion for local completability which is both necessary and sufficient for almost all configurations $P$.

Let $P \in \AR^{D \times N}$ denote a generic state-measurement configuration. Then, by Theorem~\ref{Theorem.local.rigidity}, the following holds true with probability 1.
\begin{enumerate}[(i)]
\item		$P \in U \cap G_{\Omega}^{-1}(\vec{\mathcal{K}})$,
\item		$\rank\; d(G_{\Omega})_{P} = r$,
\item		$U \cap G_{\Omega}^{-1}(\vec{\mathcal{K}})$ is a $(DN - r)$-dimensional submanifold.
\end{enumerate}
Consequently, with probability 1, the smooth $\vec{\mathcal{K}}$-compatible deformations of $P$ form locally around $P$ a submanifold whose dimension is independent of the probability measure $\mu \in \mathcal{L}$ describing $P$.\footnote{Recall~\eqref{Rigidity:the.abs.cont.assumption} from the Introduction.} This bring us to the following Corollary of Theorem~\ref{Theorem.local.rigidity}.

\begin{corollary}\label{Rigidity:corollary.about.the.universality.of.hat.d}
	Assume that $P$ is generic. Then, with probability 1
	\begin{multline}\nn
		\dim( \tilde{N}_{P} ) = \dim(U \cap G_{\Omega}^{-1}(\vec{\mathcal{K}})) \\
		\Leftrightarrow
		\exists \tilde{U} \subseteq U \text{ open neighborhood of $P$}: \; \tilde{U} \cap \tilde{N}_{P} = \tilde{U} \cap G_{\Omega}^{-1}(\vec{\mathcal{K}}).
	\end{multline}
\end{corollary}

\begin{proof}
	The direction ``$\Leftarrow$'' is obvious. It is left to explain the direction ``$\Rightarrow$''. Note that $\dim( \tilde{N}_{P} ) = \dim( \tilde{N}_{P} \cap U)$. Thus, by assumption, $\dim( \tilde{N}_{P} \cap U) =  \dim(U \cap G_{\Omega}^{-1}(\vec{\mathcal{K}}))$ implying that $\tilde{N}_{P} \cap U$ is a submanifold of $U \cap G_{\Omega}^{-1}(\vec{\mathcal{K}})$ with codimension 0. It follows that $\tilde{N}_{P} \cap U$ is open in $U \cap G_{\Omega}^{-1}(\vec{\mathcal{K}})$ (with respect to the subspace topology; see for example Proposition~5.1 in~\cite{Lee2013}). Consequently, there exists an open neighborhood $\tilde{U}$ ($\tilde{U} \subseteq U$) of $P$ such that 
	\[
		\tilde{U} \cap \bigl( \tilde{N}_{P} \cap U \bigr) = \tilde{U} \cap \bigl( G_{\Omega}^{-1}(\vec{\mathcal{K}}) \cap U \bigr)
		\Leftrightarrow
		\tilde{U} \cap \tilde{N}_{P} = \tilde{U} \cap G_{\Omega}^{-1}(\vec{\mathcal{K}}).
	\]
 This concludes the proof of the Corollary.	
\end{proof}

From Corollary~\ref{Rigidity:corollary.about.the.universality.of.hat.d} we arrive at Algorithm~\ref{flexibility.test} because $P$ is locally completable if and only if
\[
	\exists \tilde{U} \subseteq U \text{ open neighborhood of $P$}: \; \tilde{U} \cap \tilde{N}_{P} = \tilde{U} \cap G_{\Omega}^{-1}(\vec{\mathcal{K}}).
\]
Algorithm~\ref{flexibility.test} allows to test the unknown configuration $P$ for local completability with probability 1. We conclude that local completability is a universal property because either almost all state-measurement configurations $P$ are locally completable or they are not.

\begin{algorithm}
\caption{(Cucuringu, Singer~\cite{SingerCucuringu})}
\label{flexibility.test}
\begin{algorithmic}[1]
\Require $D$, $N = W + VK$, $\Omega$, property~\eqref{Rigidity:the.abs.cont.assumption} (see Section~\ref{Rigidity:generic.configurations}) is satisfied.
\State Draw $Q$ at random from the Lebesgue measure on $\AR^{D \times N}$.
\State Then, with probability 1, $\rank(d(G_{\Omega})_{Q})  = DN - \frac{1}{2} D \left( D - 1 \right)$ if and only if $P \in \AR^{D \times N}$ is locally completable.
\end{algorithmic}
\end{algorithm}

Algorithm~\ref{flexibility.test} has already been introduced in~\cite{SingerCucuringu} by M.~Cucuringu and A.~Singer. Here, we strengthened their Algorithm by proving that for generic configurations $P$ their condition (namely, $\rank(d(G_{\Omega})_{Q})  = DN - \frac{1}{2} D \left( D - 1 \right)$) is not only sufficient but also necessary for local completability.

We would like to stress the requirement that $P$ is sampled from measure satisfying~\eqref{Rigidity:the.abs.cont.assumption}. This condition is not met if (for instance) the configuration $P$ is known to carry some rank-deficient states or measurements. Under these circumstances, as proven in~\cite{SingerCucuringu}, $\rank(d(G_{\Omega})_{Q})  = DN - \frac{1}{2} D \left( D - 1 \right)$ is sufficient but not necessary for local completability.

\subsection{Completability phase diagrams}\label{Sec:drawing.phase.diagrams}

For each tupel $(W,V,D,\Omega,R)$, Algorithm~\ref{flexibility.test} determines the property `locally completable' or `locally flexible'. The results Algorithm~\ref{flexibility.test} produces when varying $W$ and $V$, can be merged to form a completability phase diagram (see for instance Figure~\ref{fig:2D.rigidity.diagram.pure.states}). In the following, when drawing completability diagrams, we are each time setting $D$ (equal to $\dim(\HR)^2$) to a specific value and start to vary the number of states $W$ and the number of measurements $V$. Points in these diagrams that are incompatible with the constraint $N \geq D$ are automatically assigned to the locally completable phase; $N$ denotes the number of columns of $P$ from Eq.~\eqref{definition.of.n.in.the.flexibility.context}.

We computed the phase diagrams for local completability for the following scenarios. In both cases we assume that the entries marking $\data$ are part of $\Omega$.
\begin{enumerate}[(I)]
\item		\emph{Scenario ``approximate pure states''}. We know a priori that the prepared states are approximately pure. Thus, our knowledge is of the form
		\beq
			\Omega_{\mathrm{st}} = 
			\left( \begin{array}{cccc}
			\bullet		& \circ		& \cdots			& \circ			 \\ 
			\circ			& \bullet		& 				& \vdots				\\ 
			\vdots		& 			& \ddots			& \circ			\\ 
			\circ			& \cdots		& \circ			& \bullet			\\ 
			\end{array} \right), \;
			\Omega_{\mathrm{m}} = \emptyset.
		\eeq
		Here, the symbol ``$\bullet$'' marks the known entries whereas the symbol ``$\circ$'' marks the unknown entries.
\item		\emph{Scenario ``approximate projective measurements with known degeneracies''}.	We know a priori that the performed measurements are approximately projective, and we know the degeneracy of each projector. Thus, our knowledge is of the form	
		\beq\label{Eq:pattern.of.known.entries.for.known.projective.measurements}
			\Omega_{\mathrm{st}} = \emptyset, \;
			\Omega_{\mathrm{m}} = 
			\left( \begin{array}{ccccccc}
			\cline{1-3}
			\cm{\bullet}	& \bullet		& \mc{\bullet}		& 			& 			& 			& 			 \\ 
			\cm{\bullet}	& \bullet		& \mc{\bullet}		& 			& 			& 			& 		 \\ 
			\cm{\bullet}	& \bullet		& \mc{\bullet}		& 			& 			& 			& 			\\ 
			\cline{1-6}
						& 			& 				& \cm{\bullet}	& \bullet		& \mc{\bullet}	& 		 \\ 
						& 			& 				& \cm{\bullet}	& \bullet		& \mc{\bullet}	& 		 \\ 
						& 			& 				& \cm{\bullet}	& \bullet		& \mc{\bullet}	& 		 \\ 
			\cline{4-6}
						& 			& 				& 			& 			& 			& \ddots  
			\end{array} \right).
		\eeq
		The blocks are of size $\dim(\HR) \times \dim(\HR)$.
\end{enumerate}
Please note the additive ``approximate'' in the names of the different scenarios. It stresses that it is important to make sure that our a priori knowledge has to correspond to \emph{generic} state-measurement configurations. For instance, if we enforce exact projective measurements, then the underlying state-measurement configuration are idealized and cannot be treated as generic configurations. Thus, when considering idealized knowledge about $G$ it is important to run the test for local completability for the considered special case. Then, passing the test is sufficient but necessary for local completability. Calculations were performed in Matlab using the rank function. Here, some caution is advised because this direct numerical computation of the rank of very large matrices is prone to numerical errors. The results are shown in Figures~\ref{fig:2D.rigidity.diagram.projective.measurements} to~\ref{fig:4D.rigidity.diagram.all.lengths.proj.non.deg}.

\begin{figure}[tbp]
\centering
\includegraphics[width=0.8\columnwidth]{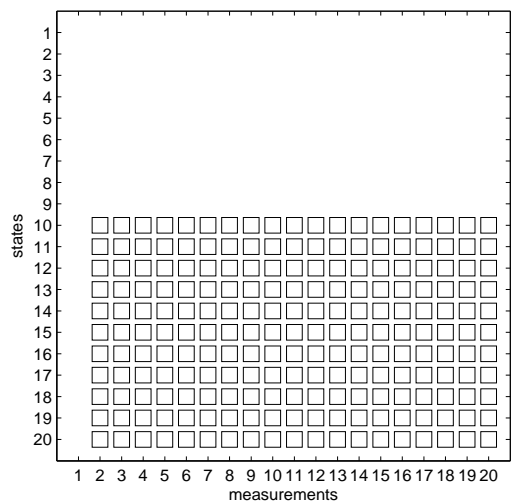}
\caption{Completability phase diagram in case of $\dim \HR = 2$ and scenario (I). The squares mark the locally completable phase, i.e., almost all configurations are locally completable.}
\label{fig:2D.rigidity.diagram.pure.states}
\end{figure}

\begin{figure}[tbp]
\centering
\includegraphics[width=0.8\columnwidth]{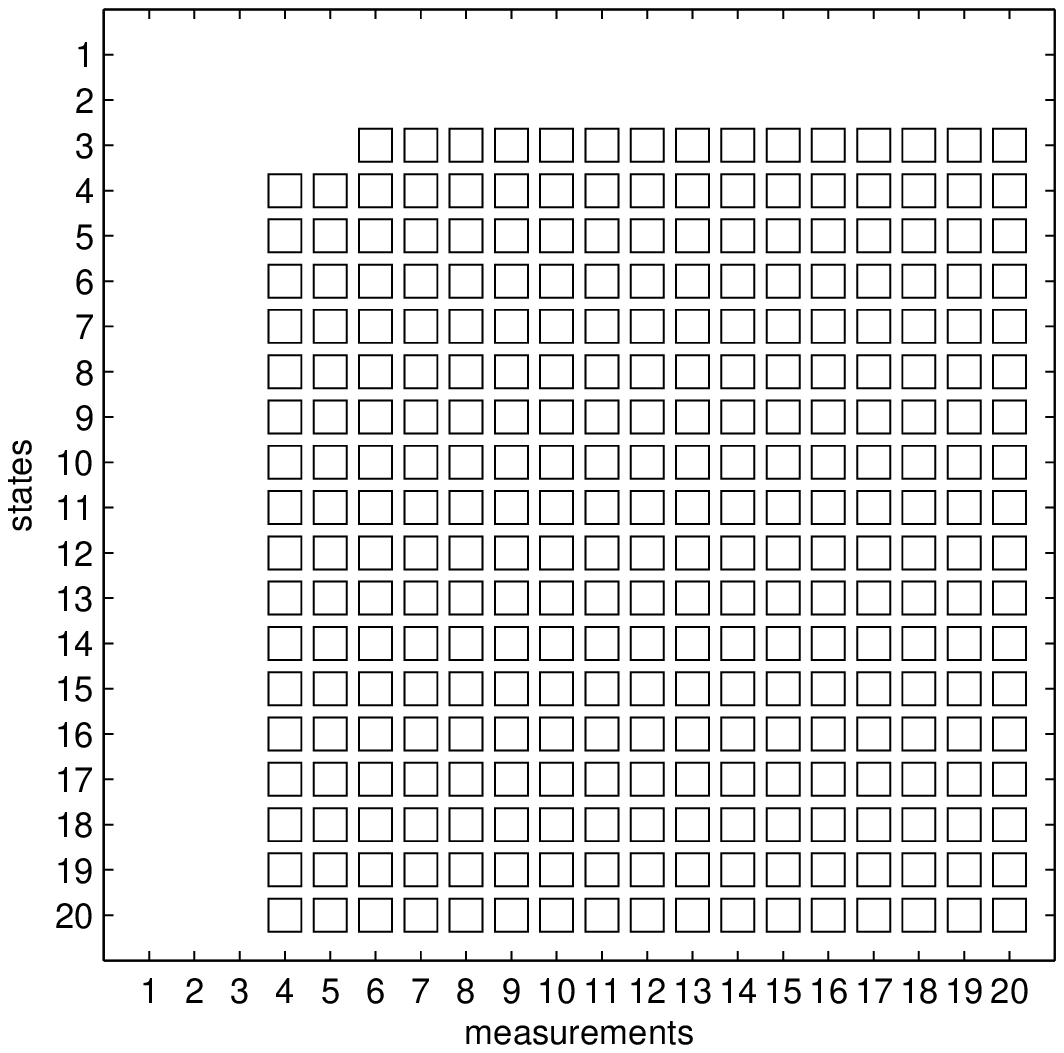}
\caption{Local completability phase diagram in case of $\dim \HR = 2$ and scenario (II). The squares mark the locally completable phase, i.e., almost all configurations are locally completable.}
\label{fig:2D.rigidity.diagram.projective.measurements}
\end{figure}

\begin{figure}[tbp]
\centering
\includegraphics[width=0.8\columnwidth]{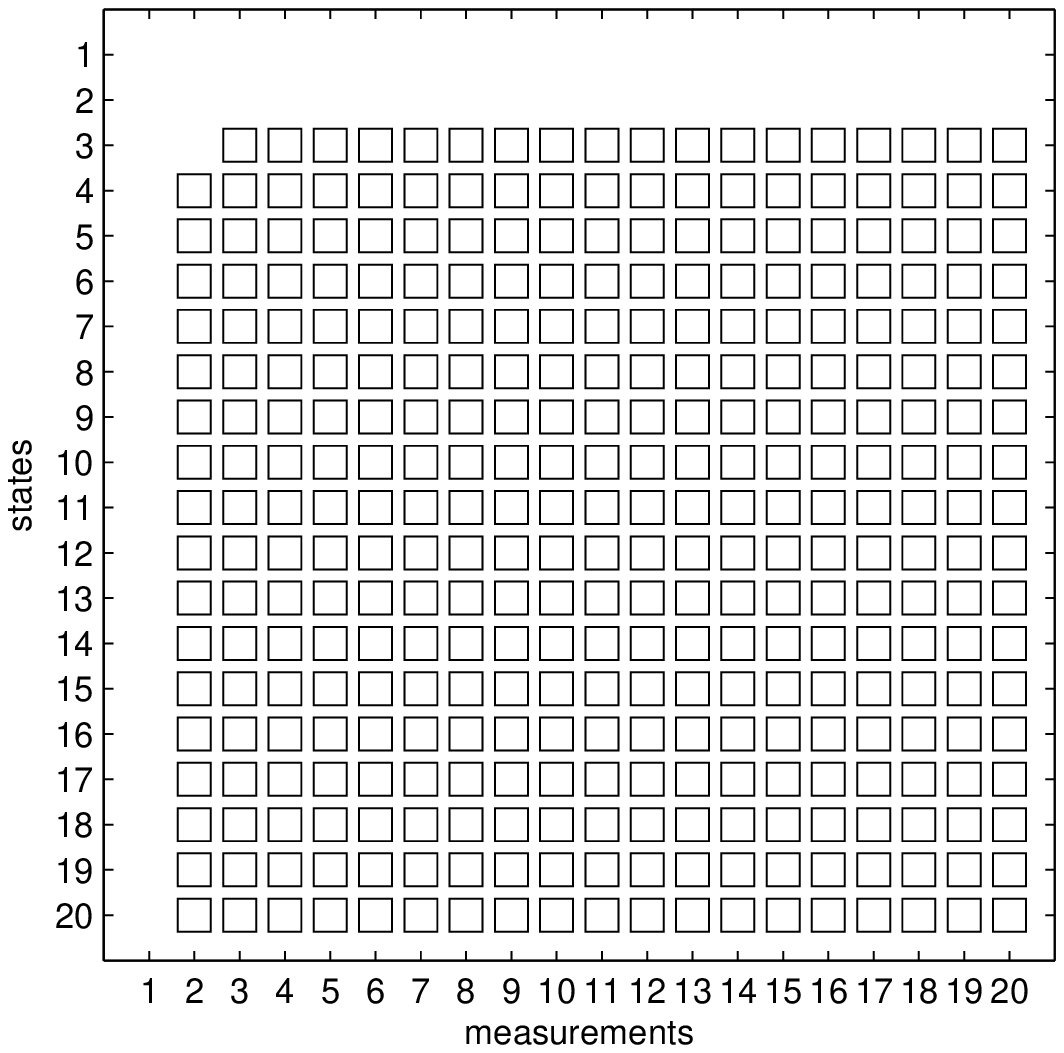}
\caption{Completability phase diagram in case of $\dim \HR = 2$ for scenario (I) in combination with scenario (II). The squares mark the locally completable phase, i.e., almost all configurations are locally completable.}
\label{fig:2D.rigidity.diagram.all.lengths.proj.non.deg}
\end{figure}

\begin{figure}[tbp]
\centering
\includegraphics[width=0.8\columnwidth]{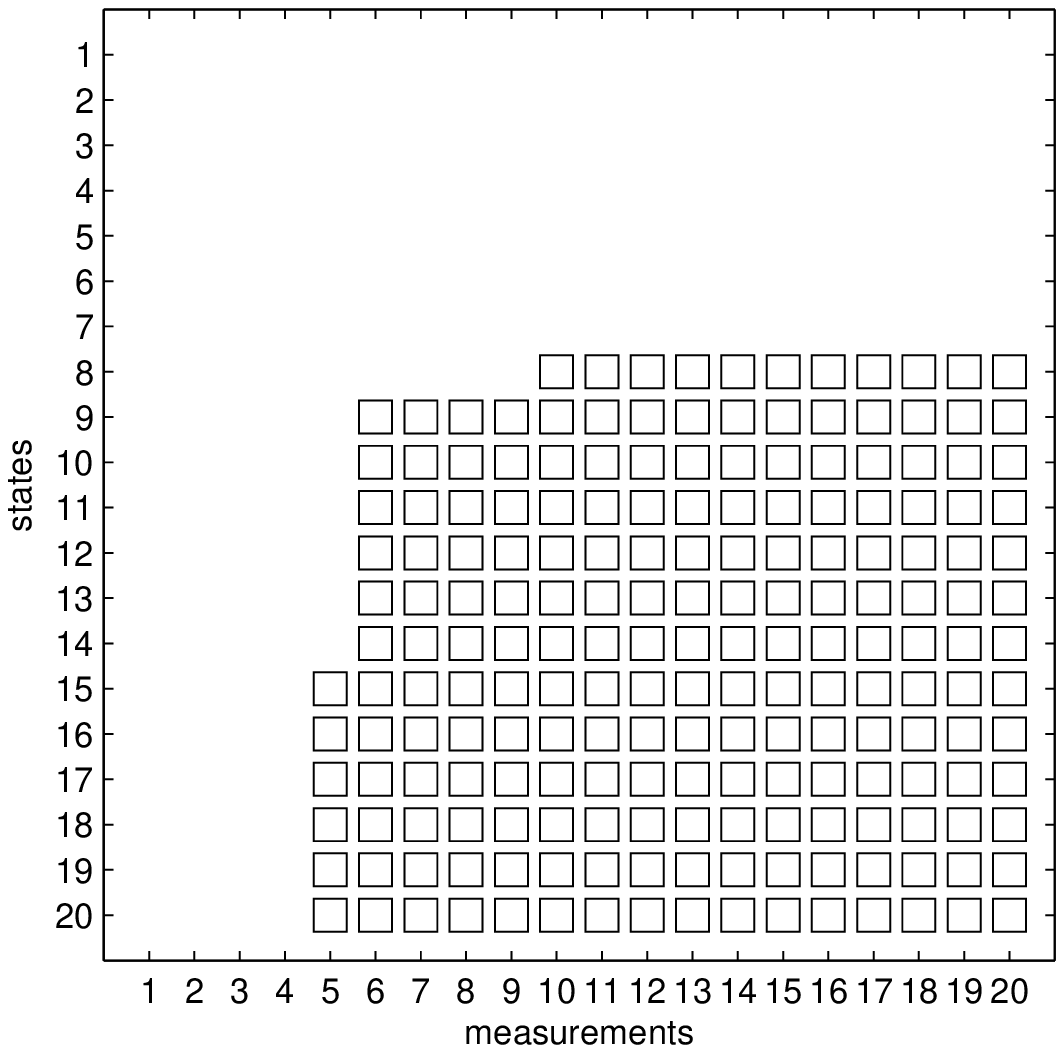}
\caption{Completability phase diagram in case of $\dim \HR = 3$ for scenario (I) in combination with scenario (II). The squares mark the locally completable phase, i.e., almost all configurations are locally completable.}
\label{fig:3D.rigidity.diagram.all.lengths.proj.non.deg}
\end{figure}

\begin{figure}[tbp]
\centering
\includegraphics[width=0.8\columnwidth]{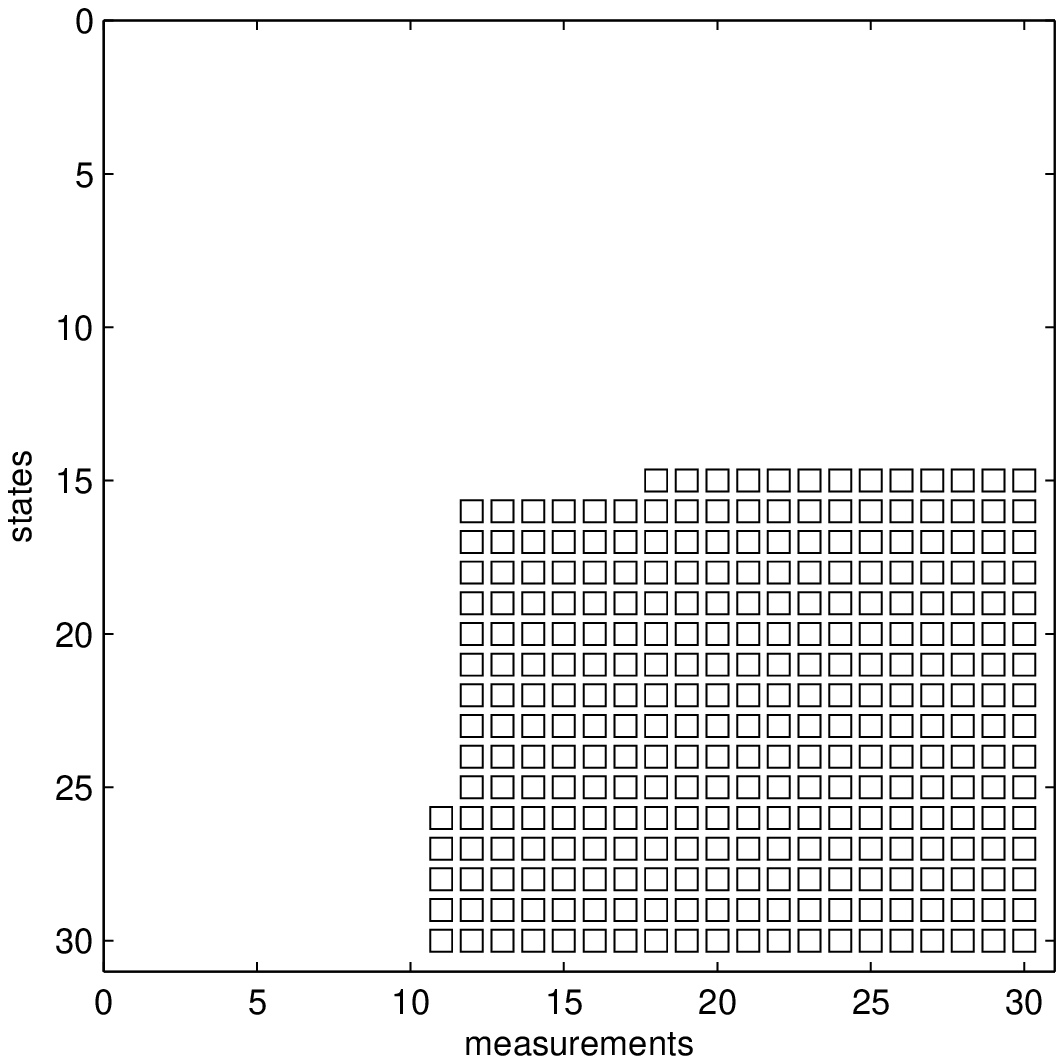}
\caption{Completability phase diagram in case of $\dim \HR = 4$ for scenario (I) in combination with scenario (II). The squares mark the locally completable phase, i.e., almost all configurations are locally completable.}
\label{fig:4D.rigidity.diagram.all.lengths.proj.non.deg}
\end{figure}

\section{Global completability}\label{Section:Global.rigidity}

The criterion we introduced in Section~\ref{Sec:sufficient.conditions} allowed to test for local completability, i.e., for the non-deformability of configurations. However, this is not completely sufficient for making sure that the state-measurement Gram matrix is uniquely specified by our knowledge $\Omega$ because there is still room for discrete symmetries. The purpose of this section is the derivation of a necessary and sufficient test to certify global completability. It is applicable for a large class of $\Omega \subseteq \{1,...,N\}^{\times 2}$. Passing this test guarantees that (up to trivial transformations) the state-measurement configuration $P$ is uniquely determined by the a priori knowledge $G_{\Omega} = \knowledge$, i.e., the state-measurement Gram matrix $G$ is uniquely determined by our knowledge $G_{\Omega} = \knowledge$. The idea of the proof is simple and intuitive geometrically: Lemma~\ref{Global.rigidity:A.lemma} states that all $\data$-compatible configurations $P$ (recall Eq.~\eqref{temp.82394002B}) can be regarded as the orbit of an action of $GL(\AR^D)$ on $\AR^{D \times N}$. Thus, we only need to determine conditions to break these orbits.

\subsection{Derivation of a necessary and sufficient criterion for global completability}

Define
\beq\begin{aligned}\label{Global.rigidity:Def.thetasharp}
	\Omega
	&=	\{ (i,j) \; | \; G_{ij} \text{ known a priori} \},\\
	\theta^{(\mathrm{st})}
	&:=	\{ (\theta^{(\mathrm{st})}(1), \theta^{(\mathrm{st})}(2),...) \; | \\& \ \ \  \ \ \ \ \theta^{(\mathrm{st})}(k) \in \{ 1,...,W \}^{\times 2}, \; G_{\theta^{(\mathrm{st})}(k)} \text{ known} \; \forall k  \},\\
	\theta^{(\mathrm{m})}
	&:=	\{ (\theta^{(\mathrm{m})}(1), \theta^{(\mathrm{m})}(2),...) \; | \\& \ \ \ \ \ \  \ \theta^{(\mathrm{m})}(k) \in \{ 1,...,VK \}^{\times 2}, \; G_{(W,W)+\theta^{(\mathrm{m})}(k)} \text{ known} \; \forall k  \},
\end{aligned}\eeq
i.e., $\theta^{(\mathrm{st})}(k)$ marks a known entry of the state-Gram matrix $G_{\mathrm{st}}$ whereas $\theta^{(\mathrm{m})}(k)$ marks a known entry of the measurement-Gram matrix $G_{\mathrm{m}}$. In this way we can use $\theta^{(\mathrm{st})}$ and $\theta^{(\mathrm{m})}$ to describe our a priori knowledge about the Gram matrices $G_{\mathrm{st}}$ and $G_{\mathrm{m}}$.

\begin{lemma}\label{Global.rigidity:existence.Pst.Pm.lemma}
	Let $\data \in \AR^{W \times VK}$ satisfying $\rank(\data) = D$. Then, there exist $\Pst^{(0)} \in \AR^{D \times W}$ and $\Pm^{(0)} \in \AR^{D \times VK}$ such that $\data = (\Pst^{(0)})^T \Pm^{(0)}$ and $\rank(\Pst^{(0)}) = \rank(\Pm^{(0)}) = D$.
\end{lemma}

\begin{proof}
	By the SVD of $\data$,
	\[
		\data = U S V^T = U \left( \begin{array}{cc}  \id_{D} & 0  \\ 0  &  0  \end{array}\right) S V^T = (\Pst^{(0)})^T \Pm^{(0)}
	\]
	with
	\[
		\Pst^{(0)} := (U_{(\cdot),1:D})^T, \; \Pm^{(0)} :=  \left( \begin{array}{ccc}  s_{1} &  &   \\    &  \ddots &   \\    &    &  s_{D} \end{array}\right) (V^T)_{1:D,(\cdot)}.
	\]
	Moreover, $\rank(\Pst^{(0)}) = \rank(\Pm^{(0)}) = D$ because
	\[
		D = \rank(\data) \leq \min\{ \rank(\Pst^{(0)}), \rank(\Pm^{(0)}) \} \leq D.
	\]
\end{proof}

\begin{lemma}\label{Global.rigidity:A.lemma}
	Let $\data \in \AR^{W \times VK}$ satisfying $\rank(\data) = D$, and let $\Pst^{(0)} \in \AR^{D \times W}$ and $\Pm^{(0)} \in \AR^{D \times VK}$ such that $\data = \bigl(\Pst^{(0)}\bigr)^T \Pm^{(0)}$ (existence guaranteed by Lemma~\ref{Global.rigidity:existence.Pst.Pm.lemma}). Assume $\Pst \in \AR^{D \times W}$ and $\Pm \in \AR^{D \times VK}$. Then, 
	\[
		 \Pst^T \Pm = \data \; \Leftrightarrow \; \exists A \in GL(D) : \Pst = A^{-T} \Pst^{(0)}, \; \Pm = A \Pm^{(0)}.
	\]
\end{lemma}

See~\cite{Stark2013UniquenessInSelfConsistentTomography} for a proof of Lemma~\ref{Global.rigidity:A.lemma}.

\begin{lemma}\label{Global.rigidity:Lemma.stating.uniqueness.of.ATA}
	Let $\data \in \AR^{W \times VK}$ satisfying $\rank(\data) = D$. Hence, there exists a $(D \times D)$ submatrix $\Dsquare$ of $\data$ such that $\rank(\Dsquare) = D$. Let $\thetasharp = (\thetasharp(j))_{j=1}^{J}$ (recall Eq.~\eqref{Global.rigidity:Def.thetasharp}) denote the index set marking the a priori known entries of $G_{\sharp}$ ($\sharp = (\mathrm{st})$ or $\sharp = (\mathrm{m})$), i.e., the entries $(G_{\sharp})_{\thetasharp(j)}$ of $G$ are known a priori for all $j \in \{ 1,...,J \}$. Let $\Pst^{(0)} \in \AR^{D \times W}$ and $\Pm^{(0)} \in \AR^{D \times VK}$ be such that $(\Pst^{(0)})^T \Pm^{(0)} = \data$ and let $P \in \AR^{D \times N}$ be such that $G := P^T P$ is compatible with our entire a priori knowledge specified by $\data$ and the entries specified by $(\thetasharp(j))_{j=1}^{J}$. Let $R^{(0)}_{\mathrm{st}}, R^{(0)}_{\mathrm{m}} \in \AR^{D \times D}$ be submatrices of $\Pst^{(0)}$ resp. $\Pm^{(0)}$ corresponding to the location of $\Dsquare$ in $\data$, i.e., $(R_{\mathrm{st}}^{(0)})^T R^{(0)}_{\mathrm{m}} = \Dsquare$. Let $A \in GL(\AR^{D})$ be such that $\Pst = A^{-T} \Pst^{(0)}$ and $\Pm = A \Pm^{(0)}$ (existence guaranteed by Lemma~\ref{Global.rigidity:A.lemma}). Then, $A^TA$ is uniquely determined by $P^{(0)}$ and $\data$ if and only if\footnote{For $Q \in \AR^{N \times N}$, $\mathrm{vec}(N) \in \AR^{N^2}$ is the vector created from $M$ by stacking all the columns of $M$ on top of each other.}
	\begin{multline}\nn
		\rank\Bigl( \mathrm{vec}(R^{(0)}_{\bar{\sharp}} B_{\theta(1)} (R^{(0)}_{\bar{\sharp}})^T) \Bigl|  \cdots \Bigl| \mathrm{vec}(R^{(0)}_{\bar{\sharp}} B_{\theta(J)} (R^{(0)}_{\bar{\sharp}})^T)  \Bigr) = \frac{1}{2} D \left( D + 1 \right).
	\end{multline}	
	Here, $\bar{\sharp} = (\mathrm{st})$ if $\sharp = (\mathrm{m})$ reps. $\bar{\sharp} = (\mathrm{m})$ if $\sharp = (\mathrm{st})$, and
	\begin{multline}\nn
		B_{\theta(n)} := \frac{1}{2} \left(  \overrightarrow{(P^{(0)}_{\sharp})}_{\thetasharp(j)_{1}}  \overrightarrow{(P^{(0)}_{\sharp})}_{{\thetasharp(j)}_{2}}^T + \overrightarrow{(P^{(0)}_{\sharp})}_{{\thetasharp(j)}_{2}}  \overrightarrow{(P^{(0)}_{\sharp})}_{{\thetasharp(j)}_{1}}^T  \right), \\ \forall j = 1, ..., J,
	\end{multline}
	where $\overrightarrow{(P^{(0)}_{j})}$ denotes the $j$-th column of the matrix $P^{(0)}$.
	
\end{lemma}

\begin{proof}
	By assumption, for $\sharp = (\mathrm{st})$
	\begin{multline}\nn
		(G_{\sharp})_{\theta^{\mathrm{st}}(j)} 
		= \SP{ (P_{\mathrm{st}})_{\theta^{\mathrm{st}}(j)_{1}} }{ (P_{\mathrm{st}})_{\theta^{\mathrm{st}}(j)_{2}} }
		= \SP{ (P^{(0)}_{\mathrm{st}})_{\theta^{\mathrm{st}}(j)_{1}} }{ (A^T A)^{-1}  (P^{(0)}_{\mathrm{st}})_{\theta^{\mathrm{st}}(j)_{2}} }, \\ j \in \{ 1,...,J \},
	\end{multline}
	and
	\begin{multline}\nn
		(G_{\sharp})_{\theta^{\mathrm{m}}(j)} 
		= \SP{ (P_{\mathrm{m}})_{\theta^{\mathrm{m}}(j)_{1}} }{ (P_{\mathrm{m}})_{\theta^{\mathrm{m}}(j)_{2}} }
		= \SP{ (P^{(0)}_{\mathrm{m}})_{\theta^{\mathrm{m}}(j)_{1}} }{ A^T A  (P^{(0)}_{\mathrm{m}})_{\theta^{\mathrm{m}}(j)_{2}} }, \\ j \in \{ 1,...,J \},
	\end{multline}
	for $\sharp = (\mathrm{m})$ fixed by our a priori knowledge. Set
	\[
		M := \left\{ \begin{array}{ll}   
		(A^T A)^{-1}, 	&\text{ if $\sharp = (\mathrm{st})$,}\\
		A^T A,	 	&\text{ if $\sharp = (\mathrm{m})$.}
		\end{array} \right.
	\]
	So $M$ is in both cases a real, symmetric matrix (positive definite) satisfying
	\beq
		G_{\thetasharp(j)}
		= \tr\Bigl( \overrightarrow{(P^{(0)}_{\sharp})}_{\thetasharp(j)_{2}}   \overrightarrow{(P^{(0)}_{\sharp})}^T_{\thetasharp(j)_{1}} M \Bigr), \; \forall j \in \{ 1,...,J \}.
	\eeq
	Note that this property is equivalent to
	\beq\label{Global.rigidity:imposed.onstraints}
		G_{\thetasharp(j)}
		=	\tr\bigl(  B_{\theta(j)}  M \bigr), \; \forall j \in \{ 1,...,J \},
	\eeq
	because of the invariance of the trace under transposition and because $M^T = M$. Note that the linear constraints~\eqref{Global.rigidity:imposed.onstraints} specify the symmetric matrix $M$ uniquely if and only if
	\beq\label{Global.rigidity:criterion.temp.5784579}
		\mathrm{span} \left\{  B_{\theta(j)}  \right\}_{j=1}^J
		= \Sym( \AR^{D} ).
	\eeq
	Note that $R_{\bar{\sharp}} \in GL(\AR^{D})$ because $D = \rank( \Dsquare ) \leq \rank(R_{\bar{\sharp}}) \leq D$. Define $\mathcal{R} : \Sym( \AR^{D} ) \rightarrow \Sym( \AR^{D} )$ by
	\[
		\mathcal{R}: S \mapsto R_{\bar{\sharp}} S R_{\bar{\sharp}}^T.
	\]
	The invertibility of $R_{\bar{\sharp}}$ implies $\mathcal{R} \in GL(\Sym( \AR^{D} ))$. Hence, by the criterion~\eqref{Global.rigidity:criterion.temp.5784579}, the linear constraints ~\eqref{Global.rigidity:imposed.onstraints} specify $M$ uniquely if and only if
	\beq\label{Global.rigidity:criterion.temp.5784579B}
		\mathrm{span} \left\{  R_{\bar{\sharp}} B_{\theta(j)} R_{\bar{\sharp}}^T  \right\}_{j=1}^J
		= \Sym( \AR^{D} ).
	\eeq
	Eq.~\eqref{Global.rigidity:criterion.temp.5784579B} is satisfied if and only if
	\begin{multline}
		\rank\Bigl( \mathrm{vec}(R_{\bar{\sharp}} B_{\theta(1)} R_{\bar{\sharp}}^T) \Bigl| \mathrm{vec}(R_{\bar{\sharp}} B_{\theta(2)} R_{\bar{\sharp}}^T) \Bigl| \cdots  \Bigr)\\
		= \dim(\Sym( \AR^{D} )) = \frac{1}{2} D \left( D + 1 \right).
	\end{multline}
	 This concludes the proof of the Lemma.
	
\end{proof}

\begin{theorem}\label{Global.rigidity:Theorem.stating.the.sufficient.criterion}
	Let $\data \in \AR^{W \times VK}$ satisfying $\rank(\data) = D$. We assume that the states and POVM elements have been renamed such that\footnote{$\data_{1:D,1:D}$ refers to the top left submatrix of $\data$ of size $D \times D$.} $\rank(\data_{1:D,1:D}) = D$. Moreover, let $\thetasharp = (\thetasharp(j))_{j=1}^{J}$ (recall Eq.~\eqref{Global.rigidity:Def.thetasharp}) be an index set such that we know a priori the value of $(G_{\sharp})_{\thetasharp(j)}$ for all $j \in \{ 1,...,J \}$ where $\sharp = (\mathrm{st})$ or $\sharp = (\mathrm{m})$. Define $\mathcal{M} := \bigl( \mathrm{vec}\bigl( \mathcal{N}_{1}^{(\sharp)} \bigr) \bigl| \cdots \bigl| \mathrm{vec}\bigl( \mathcal{N}_{J}^{(\sharp)} \bigr) \bigr)$ where
	\begin{multline}\nn
		\mathcal{N}_{j}^{(\mathrm{st})}
		:=
		\frac{1}{2} \Bigl(  \bigl( \data_{\theta^{(\mathrm{st})}(j)_{1},1:D} \bigr)^T \data_{\theta^{(\mathrm{st})}(j)_{2},1:D}  + \bigl( \data_{\theta^{(\mathrm{st})}(j)_{2},1:D} \bigr)^T \data_{\theta^{(\mathrm{st})}(j)_{1},1:D} \Bigr),
	\end{multline}
	and
	\begin{multline}\nn
		\mathcal{N}_{j}^{(\mathrm{m})}
		:=
		\frac{1}{2} \Bigl(   \data_{1:D,\theta^{(\mathrm{m})}(j)_{1}}  \bigl(\data_{1:D,\theta^{(\mathrm{m})}(j)_{2}}\bigr)^T  + \data_{1:D,\theta^{(\mathrm{m})}(j)_{2}}  \bigl(\data_{1:D,\theta^{(\mathrm{m})}(j)_{1}}\bigr)^T \Bigr).
	\end{multline}
	Then, the following are equivalent:
	\begin{enumerate}[(i)]
		\item There exists only one state-measurement Gram matrix $G$ with $\rank(G) = D$ which is compatible with our a priori knowledge.
		\item $\rank(\mathcal{M}) = D(D + 1)/2$
	\end{enumerate}
\end{theorem}

\begin{proof}
	We first show ``$\Leftarrow$'': Assume $P = (\Pst | \Pm)$, $\mathcal{P} = (\curlyPst | \curlyPm)$ are two arbitrary configurations satisfying individually all the conditions of the Theorem. Define
	\[
		G := P^TP, \; \mathcal{G} := \mathcal{P}^T \mathcal{P}.
	\]
	Let $(\Pst^{(0)} | \Pm^{(0)})$ be any configuration satisfying the data-only constraint $(\Pst^{(0)})^T \Pm^{(0)} = \data$. By Lemma~\ref{Global.rigidity:A.lemma}
	\beq\begin{aligned}\nn
		&\exists A \in GL(\AR^{D}): \Pst = A^{-T} \Pst^{(0)}, \; \Pm = A \Pm^{(0)},\\
		&\exists \mathcal{A} \in GL(\AR^{D}): \Pst = \mathcal{A}^{-T} \Pst^{(0)}, \; \Pm = \mathcal{A} \Pm^{(0)},\\
	\end{aligned}\eeq
	By Lemma~\ref{Global.rigidity:Lemma.stating.uniqueness.of.ATA},
	\[
		A^TA = \mathcal{A}^T\mathcal{A}. 
	\]
	Therefore (recall that $A^TA = \mathcal{A}^T\mathcal{A} > 0$), there exists an orthogonal matrix $O \in O(\AR^{D})$ such that $\mathcal{A} = OA$. It follows that 
	\beq\begin{aligned}\nn
		\curlyPst &= \mathcal{A}^{-T} \Pst^{(0)} = O A^{-T} \Pst^{(0)} = O \Pst\\
		\curlyPm &= \mathcal{A} \Pm^{(0)} = O A \Pm^{(0)} = O \Pm 
	\end{aligned}\eeq
	and therefore, $\mathcal{P} = O P$ implying that
	\[
		\mathcal{G} = \mathcal{P}^T \mathcal{P} = P^T O^T O P = G.
	\]
	This holds true any two state-measurement configurations $\mathcal{P},P$ satisfying the conditions from the Theorem. Hence, the conditions from the Theorem uniquely specify the state-measurement Gram matrix $G$.
	
	Consider the opposite direction ``$\Rightarrow$'': By assumption, $G$ is determined uniquely, and $\rank(G) = D$. Let $P_{1},P_{2} \in \AR^{D \times N}$. Then, $G = P_{1}^T P_{1} = P_{2}^T P_{2}$ if and only if there exists an orthogonal matrix $O \in O(\AR^{D})$ such that $P_{2} = O P_{1}$. By Lemma~\ref{Global.rigidity:A.lemma}, there exist $A_{1},A_{2} \in GL(\AR^{D})$ such that $P_{\mathrm{st},1} = A_{1}^{-T} \Pst^{(0)}$, $P_{\mathrm{m},1} = A_{1} \Pm^{(0)}$, $P_{\mathrm{st},2} = A_{2}^{-T} \Pst^{(0)}$, and $P_{\mathrm{m},2} = A_{2} \Pm^{(0)}$. It follows that
	\[
		A_{2} \Pm^{(0)} = P_{\mathrm{m},2} = O P_{\mathrm{m},1} = O A_{1} \Pm^{(0)}.
	\]
	Since $\Pm^{(0)}$ is full-rank, it follows that $A_{2} = OA_{1}$ and consequently, $A_{2}^T A_{2} = A_{1}^T A_{1}$ is determined uniquely. Hence, by Lemma~\ref{Global.rigidity:Lemma.stating.uniqueness.of.ATA}, $\rank(\mathcal{M}) = D(D +1)/2$ (recall that the assumptions from the Theorem, we can choose $\data_{\square} = \data_{1:D,1:D}$). This concludes the proof of the Theorem.

\end{proof}

\subsection{Universality}

Let $\mathcal{M}$ be defined as in Theorem~\ref{Global.rigidity:Theorem.stating.the.sufficient.criterion}. The matrix $\mathcal{M}$ is a function of the data $\data = \Pst^T \Pm$ and can therefore be regarded as a function $\mathcal{M}(P)$ of the underlying the state-measurement configuration $P = (\Pst | \Pm)$.

\begin{theorem}\label{Global.rigidity:theorem.stating.universality}
	Fix $\Omega$ and let $\mathcal{M}(P)$ be defined as in Theorem~\ref{Global.rigidity:Theorem.stating.the.sufficient.criterion}. Then, either all $P \in \AR^{D \times N}$ are such that $\rank( \mathcal{M}(P) ) < D(D + 1)/2$, or almost all $P \in \AR^{D \times N}$ satisfy $\rank( \mathcal{M}(P) ) = D(D + 1)/2$.	
\end{theorem}

\begin{proof}

	Let 
	\[
		\mathcal{R}_{\kappa} := \{ P \in \AR^{D \times N} \; | \; \rank \left(  \mathcal{M}(P)  \right) = \kappa  \},
	\]
	define
	\[
		\omega_{\kappa}(P) := \sum_{\substack{A \subseteq  \mathcal{M}(P) \\ A \in \mathbb{R}^{\kappa \times \kappa}  }} \left( \det A \right)^2
	\]
	($A \subseteq \mathcal{M}(P)$ means that $A$ is a submatrix of $\mathcal{M}(P)$), and 
	\[
		\mathcal{N}_{\kappa} := \{ P \in \AR^{D \times N} \; | \; \omega_{\kappa}(P) = 0 \}.
	\]
	Recall that $\mathcal{N}_{j}^{(\sharp)} \in \Sym(\AR^{D})$ (see Theorem~\ref{Global.rigidity:Theorem.stating.the.sufficient.criterion}). Therefore, 
	\[
		\rank( \mathcal{M}(P) ) \leq \dim(\Sym(\AR^{D})) = \frac{1}{2} D (D + 1).
	\]
	It follows that 
	\[
		\mathcal{R}_{\kappa} = \left\{ \begin{array}{ll}  \mathcal{N}_{\kappa+1} \cap \Bigl( \AR^{D \times N} \setminus \mathcal{N}_{\kappa} \Bigr) & \forall \kappa < D(D+1)/2,   \\    \AR^{D \times N} \setminus \mathcal{N}_{D(D+1)/2} & \forall \kappa = D(D+1)/2 .  \end{array} \right. 
	\]
	Note that $\omega_{\kappa}(P)$ is a polynomial in the entries of $P$ because $\omega_{\kappa}(P)$ is a polynomial in the entries of $\mathcal{M}(P)$, the entries of $\mathcal{M}(P)$ are polynomial in the entries of $\data$, and $\data = \Pst^T \Pm$ is polynomial in the entries of $P$. Hence, $\mathcal{N}_{\kappa}$ is an algebraic set. Analogous to the proof of Lemma~\ref{N.closed.zero-set.lemma} in Section~\ref{Section:appendix} it follows that either $\mathcal{N}_{D(D+1)/2}$ is a zero set, or $\mathcal{N}_{D(D+1)/2} = \AR^{D \times N}$. Assume that $\mathcal{N}_{D(D+1)/2}$ is a zero set. Then, $\mathcal{R}_{D(D+1)/2} = \AR^{D \times N} \setminus \mathcal{N}_{D(D+1)/2}$ is full measure. On the other hand, if $\mathcal{N}_{D(D+1)/2} = \AR^{D \times N}$ then, $\mathcal{R}_{D(D+1)/2} = \AR^{D \times N} \setminus \mathcal{N}_{D(D+1)/2} = \emptyset$.	 This concludes the proof of the Theorem.
\end{proof}

By Theorem~\ref{Global.rigidity:theorem.stating.universality}, almost all state-measurement configurations $P \in \AR^{D \times (W+VK)}$ lead to the same answer to the question whether or not $\rank( \mathcal{M}(P) ) = D(D + 1)/2$. In this sense, this is a universal property. Assume that~\eqref{Rigidity:the.abs.cont.assumption}  applies (see Section~\ref{Rigidity:generic.configurations}). Then, Theorem~\ref{Global.rigidity:Theorem.stating.the.sufficient.criterion} and Theorem~\ref{Global.rigidity:theorem.stating.universality} imply for fixed knowledge $\Omega$ that either
\beq\label{Global.rigidity:alternative.A}
	\rank( \mathcal{M}(P) ) = D(D + 1)/2 \text{ with probability 1,}
\eeq
or
\beq\label{Global.rigidity:alternative.B}
	\rank( \mathcal{M}(P) ) \neq D(D + 1)/2 \text{ with probability 1.}
\eeq
Which of the two alternatives~\eqref{Global.rigidity:alternative.A} and~\eqref{Global.rigidity:alternative.B} applies for $\Omega$ can be tested by sampling $Q \in \AR^{D \times (W+VK)}$ at random from the Lebesgue measure. If $\rank( \mathcal{M}(Q) ) = D(D + 1)/2$ then (with probability 1) $\rank( \mathcal{M}(P) ) = D(D + 1)/2$ for almost all $P \in \AR^{D \times (W+VK)}$. On the other hand, if $\rank( \mathcal{M}(Q) ) \neq D(D + 1)/2$ then (with probability 1) $\rank( \mathcal{M}(P) ) \neq D(D + 1)/2$ for almost all $P \in \AR^{D \times (W+VK)}$. This procedure yields Algorithm~\ref{Global.rigidity:Algorithm.for.a.priori.test}.

\begin{algorithm}
\caption{(Test for generic global completability)}
\label{Global.rigidity:Algorithm.for.a.priori.test}
\begin{algorithmic}[1]
\Require $D$, $W$, $V$, $K$, $\Omega$, property~\eqref{Rigidity:the.abs.cont.assumption} is satisfied.
\State Draw $Q \in \AR^{D \times (W+VK)}$ at random with respect to probability measure which is Lebesgue absolutely continuous.
\State If $\rank( \mathcal{M}(Q) ) = D(D + 1)/2$ then (Lebesgue) almost all $P \in \AR^{D \times N}$ are globally completable.
\end{algorithmic}
\end{algorithm}

We perform Algorithm~\ref{Global.rigidity:Algorithm.for.a.priori.test} to analyze the scenarios listed below. In all cases we assume that the entries marking $\data$ are part of $\Omega$, i.e., assumed to be known.
\begin{enumerate}[(I)]
\item		\emph{Scenario ``approximate pure states''}. We know a priori that the prepared states are approximately pure. Thus, our knowledge is of the form
		\beq
			\Omega_{\mathrm{st}} = 
			\left( \begin{array}{cccc}
			\bullet		& \circ		& \cdots			& \circ			 \\ 
			\circ			& \bullet		& 				& \vdots				\\ 
			\vdots		& 			& \ddots			& \circ			\\ 
			\circ			& \cdots		& \circ			& \bullet			\\ 
			\end{array} \right), \;
			\Omega_{\mathrm{m}} = \emptyset.
		\eeq
		Here, the symbol ``$\bullet$'' marks the known entries whereas the symbol ``$\circ$'' marks the unknown entries.
\item		\emph{Scenario ``approximate projective measurements with known degeneracies''}.	We know a priori that the performed measurements are approximately projective, and we know the degeneracy of each projector. Thus, our knowledge is of the form	
		\beq\label{Eq:pattern.of.known.entries.for.known.projective.measurements}
			\Omega_{\mathrm{st}} = \emptyset, \;
			\Omega_{\mathrm{m}} = 
			\left( \begin{array}{ccccccc}
			\cline{1-3}
			\cm{\bullet}	& \bullet		& \mc{\bullet}		& 			& 			& 			& 			 \\ 
			\cm{\bullet}	& \bullet		& \mc{\bullet}		& 			& 			& 			& 		 \\ 
			\cm{\bullet}	& \bullet		& \mc{\bullet}		& 			& 			& 			& 			\\ 
			\cline{1-6}
						& 			& 				& \cm{\bullet}	& \bullet		& \mc{\bullet}	& 		 \\ 
						& 			& 				& \cm{\bullet}	& \bullet		& \mc{\bullet}	& 		 \\ 
						& 			& 				& \cm{\bullet}	& \bullet		& \mc{\bullet}	& 		 \\ 
			\cline{4-6}
						& 			& 				& 			& 			& 			& \ddots  
			\end{array} \right).
		\eeq
		The blocks are of size $\dim(\HR) \times \dim(\HR)$.
\item		\emph{Scenario ``approximate projective measurements with unknown degeneracies''}. We know a priori that the performed measurements are approximately projective, but we do not know the degeneracies of the projectors. Thus, our knowledge is of the form				
		\beq
			\Omega_{\mathrm{st}} = \emptyset, \;
			\Omega_{\mathrm{m}} = 
			\left( \begin{array}{ccccccc}
			\cline{1-3}
			\cm{\circ}		& \bullet		& \mc{\bullet}		& 			& 			& 			& 			 \\ 
			\cm{\bullet}	& \circ		& \mc{\bullet}		& 			& 			& 			& 		 \\ 
			\cm{\bullet}	& \bullet		& \mc{\circ}		& 			& 			& 			& 			\\ 
			\cline{1-6}
						& 			& 				& \cm{\circ}	& \bullet		& \mc{\bullet}	& 		 \\ 
						& 			& 				& \cm{\bullet}	& \circ		& \mc{\bullet}	& 		 \\ 
						& 			& 				& \cm{\bullet}	& \bullet		& \mc{\circ}	& 		 \\ 
			\cline{4-6}
						& 			& 				& 			& 			& 			& \ddots  
			\end{array} \right).
		\eeq
		The blocks are of size $\dim(\HR) \times \dim(\HR)$.
\end{enumerate}
As in the discussion of local completability, we would like to emphasize the additive ``approximate'' in the names of the different scenarios. It stresses that it is crucial to make sure that our a priori knowledge has to correspond to \emph{generic} state-measurement configurations. For instance, if we enforce exact projective measurements, then the underlying state-measurement configuration are idealized and cannot be treated as generic configurations. Thus, when considering idealized knowledge about $G$ it is important to run the test for global completability for the considered special case. In fact, the scenarios ``approximate pure states'' and ``approximate projective measurements with unknown degeneracies'' tend to become unstable when we enforce perfect pureness and perfect projectiveness respectively.

We arrive at the Tables~\ref{Table:Global.rigidity:Scenario.pure.states} to~\ref{Table:Global.rigidity:Scenario.projective.measurements.with.unknown.degeneracies} summarizing our findings for the different scenarios. Calculations were performed in Matlab using the rank function. Here, some caution is advised because this direct numerical computation of the rank of very large matrices is prone to numerical errors.

\begin{table}[htdp]
\caption{Scenario (I), i.e.,  ``\emph{approximate pure states}''. Here, $K=\dim(\HR)$.}
\begin{center}
\begin{tabular}{|c|c|c|c|}
\hline  
$\dim(\HR)$		& global completability?			& $W$			& $V$		 \\
\hline \hline
2				& yes					& 10					& 4			\\  
3				& yes					& 45					& 5			\\
4				& yes					& 136				& 6				\\
5				& yes					& 325				& 7				\\
6				& yes					& 666				& 8				\\
\hline
\end{tabular}
\end{center}
\label{Table:Global.rigidity:Scenario.pure.states}
\end{table}%

\begin{table}[htdp]
\caption{Scenario (II), i.e., ``\emph{approximate projective measurements with known degeneracies}''. Here, $K=\dim(\HR)$.}
\begin{center}
\begin{tabular}{|c|c|c|c|}
\hline  
$\dim(\HR)$		& global completability?			& $W$			& $V$ 		 \\
\hline \hline
2				& yes					& 4					& 10				\\
3				& yes					& 9					& 15				\\
4				& yes					& 16					& 23					\\
5				& yes					& 25					& 33					\\
6				& yes					& 36					& 45					\\
\hline
\end{tabular}
\end{center}
\label{Table:Global.rigidity:Scenario.projective.measurements.with.known.degeneracies}
\end{table}%

\begin{table}[htdp]
\caption{Scenario (III), i.e., ``\emph{approximate projective measurements with unknown degeneracies}''. Here, $K=\dim(\HR)$.}
\begin{center}
\begin{tabular}{|c|c|c|c|}
\hline  
$\dim(\HR)$		& global completability?				& $W$			& $V$		 \\
\hline \hline
3				& yes						& 9					& 45				\\
4				& yes						& 16					& 46					\\
5				& yes						& 25					& 55					\\
6				& yes						& 36					& 67					\\
\hline
\end{tabular}
\end{center}
\label{Table:Global.rigidity:Scenario.projective.measurements.with.unknown.degeneracies}
\end{table}%

\section{Conclusions}\label{Sec:Conclusions}

We investigated the completability of vectors $P = (\vec{p}_{1}| \cdots |\vec{p}_{N}) \in \AR^{D \times N}$ fixing the value of pairwise inner products $\vec{p}_{i}^T \vec{p}_{j}$ for some $i,j \in \{ 1, ..., N \}$. In the first part of the paper we derived a Theorem analogous to the so called Asimow-Roth Theorem from conventional rigidity theory. We saw that local completability of vectors with inner product constraints is a universal feature in the sense that either almost all state-measurement configurations $P$ are locally completable or they are not. This lead to Algorithm~\ref{flexibility.test} to probe whether or not (for given $\Omega$) almost all configurations $P$ are locally completable. Moreover, the universal nature of local completability allowed for the computation of completability phase diagrams (see Figures~\ref{fig:2D.rigidity.diagram.pure.states} to~\ref{fig:4D.rigidity.diagram.all.lengths.proj.non.deg}).

Global completability is more difficult to analyze than local completability because we cannot work in local neighborhoods of configurations $P$. However, in the practically relevant special cases where we either know something about $G_{\mathrm{st}}$ or $G_{\mathrm{m}}$, global completability becomes easy to investigate. For these circumstances we have derived a necessary and sufficient condition for global completability. It is an important open problem to determine how simultaneous knowledge about $G_{\mathrm{st}}$ and $G_{\mathrm{m}}$ can be taken into account. In contrast to the criterion derived by R.~Connelly and S.~J.~Gortler \emph{et al.} in the context of conventional rigidity theory (see~\cite{connelly2005generic,gortler2010characterizing}), this condition is not only applicable to generic configurations $P$. We showed that satisfaction respectively violation of our condition is universal in the sense that either almost all configurations $P$ satisfy or violate the condition for global completability. As before, this lead to the randomized Algorithm~\ref{Global.rigidity:Algorithm.for.a.priori.test} to probe generic global completability.

\subsection*{Acknowledgements}I want to thank Johan {\AA}berg for his continuous support throughout my time as a PhD student. I count myself very lucky having had the opportunity for discussions with Matthias Baur, Matthias Christandl, Dejan Dukaric, Hamza and Omar Fawzi, Fr\'ed\'eric Dupuis, Philippe Faist, David Gross, Patrick Pletscher, Renato Renner, Lars Steffen, L\'idia del Rio, David Sutter, Michael Walter, Sharon Wulff, and M\'ario Ziman. I acknowledge support from the Swiss National Science Foundation through the National Centre of Competence in Research ``Quantum Science and Technology''.

%
%

\section*{Appendix: Proof of Theorem~\ref{Theorem.local.rigidity}}\label{Section:appendix}

\begin{lemma}\label{Appendix:Lemma:N.tilde.is.submanifold}
	Let $P \in \AR^{DÊ\times N}$ such that $\rank(P) = D$. Then, $\tilde{N}_{P}$ is a smooth $\frac{1}{2}D(D-1)$-dimensional submanifold of $\AR^{DÊ\times N}$.
\end{lemma}

\begin{proof}
	Denote by $\mathcal{F}$ the set of full-rank matrices in $\AR^{D \times N}$. Hence, $P \in \mathcal{F}$. To prove that $\tilde{N}_{P}$ is a submanifold we are going to apply Lemma~\ref{Lee.proposition.21.7} below in the form $G \mapsto O(D)$, $M \mapsto \mathcal{F}$. For that purpose we need to check that all the conditions for the application of Lemma~\ref{Lee.proposition.21.7} are satisfied. First, we show that $\mathcal{F}$ is a smooth manifold by showing that $\mathcal{F}$ is an open set in $\AR^{D \times N}$ (see Proposition~5.3 in~\cite{Lee2013}). Then, we will prove that the stabilizers $O(D)_{P}$ (see Lemma~\ref{Lee.proposition.21.7}) in the orthogonal group are trivial, i.e., $O(D)_{P} = \{ \id \}$ for all $P \in \mathcal{F}$.
	
	The rank of a matrix is equal to the largest non-zero minor. Therefore, $Q \in \mathcal{F}$ if and only if there exists a submatrix $A \in \AR^{D \times D}$ of $Q$ such that $\det(A) \neq 0$. Enumerate all possible $(D \times D)$-submatrices of $Q \in \AR^{D \times N}$ by the index $\kappa = 1,2,...$, set $Q_{\kappa} \in \AR^{D \times D}$ to be the restriction of $Q$ to the submatrix index by $\kappa$, and define
	\[
		\mathcal{M}_{\kappa} := \{ Q \in \AR^{D \times N} \; | \; \det(Q_{\kappa}) \neq 0 \}.
	\]
	The determinant is a continuous function and $\{ 0 \}$ is closed in $\AR$. Hence, all the sets $\mathcal{M}_{\kappa}$ are open sets because their complements are closed. The full-rank set is the union of all the sets $\mathcal{M}_{\kappa}$ and therefore open.
	
	Let $P = USV^*$ be the singular value decomposition of $P$ ($U \in O(D)$, $V \in O(N)$). Let $T \in O(D)$ be an orthogonal matrix. To prove $O(D)_{P} = \{ \id \}$ we show that $TP = P \; \Rightarrow T = \id$. The l.h.s is equivalent to $TUSV^* = USV^*$. Set $\hat{S} \in \AR^{D \times D}$ such that $S = [ \hat{S} | 0_{D \times (N-D)} ]$. Note that $S V^* = \hat{S} \hat{V}^*$ where $\hat{V}^* \in \AR^{D \times N}$ consists out of the first $D$ columns of $V^*$. Therefore,
	\beq\begin{aligned}\nn
		&TUSV^* = USV^*
		\Rightarrow \
		TU\hat{S} \hat{V}^* = U\hat{S} \hat{V}^*\\
		&\Rightarrow \
		TU\hat{S} = U\hat{S}
		\Rightarrow \
		T = \id.
	\end{aligned}\eeq
	 The second implication is a consequence of the right-action of $(\hat{V}^*)^*$. Note that $\hat{S} \in GL(\AR^{D})$ because $P$ is full rank. Thus, the last implication follows from the right-action of $(U\hat{S})^{-1}$. Consequently, Lemma~\ref{Lee.proposition.21.7} can be applied because the orthogonal group $O(D)$ is closed (it is the level set associated to $\id$ of the map $A \mapsto A^*A$) and bounded, and therefore compact.
	 
	 It is left to show that $\dim(\tilde{N}_{P}) = \frac{1}{2}D(D-1)$. The dimension of the manifold $\tilde{N}_{P}$ is equal to the dimension of its tangent space $T_{P}(\tilde{N}_{P})$ at $P$,
\[
	T_{P}(\tilde{N}_{P}) = \mathrm{span} \{ \dot{Q}(0) \; | \; Q(t) = \exp(\vec{\alpha(t)} \cdot \vec{I}) P \},
\]
with $\vec{\alpha}(\cdot) \subset \AR^{\dim(SO(D))}, \vec{\alpha}(0) = 0$ arbitrary. Here, $\vec{I}$ carries all the generators of $SO(D)$. Rewriting $P = U S V^*$, 
\beq\label{S.from.SVD.of.P}
	S = \left( \begin{array}{lll | lll}  s_{1} &  &  &  0  & \cdots  & 0 \\ & \ddots &   & \vdots & \ddots  & \vdots \\ &  & s_{R} & 0 & \cdots  & 0  \end{array} \right) \in \AR^{D \times N},
\eeq
in terms of its singular value decomposition and using that the product of orthogonal matrices is an orthogonal matrix, we arrive at
\[
	T_{P}(\tilde{N}_{P}) = \mathrm{span} \{ \dot{W}(0) \; | \; W(t) = \exp(\vec{\beta(t)} \cdot \vec{I}) SV^* \},
\]
with $\vec{\beta}(\cdot) \subset \AR^{\dim(SO(D))}, \vec{\beta}(0) = 0$ arbitrary. Taking the derivative, we get $\dot{W}(0) = (\dot{\vec{\beta}}(0) \cdot \vec{I}) P$. The matrix $\dot{\vec{\beta}}(0) \cdot \vec{I}$ is an arbitrary skew-symmetric matrix because $\dot{\vec{\beta}}(0) \in  \AR^{\dim(SO(D))}$ can be arbitrary. Denote the linear subspace of skew-symmetric $(D\times D)$-matrices by $\mathcal{A} \subset \AR^{D \times D}$. We conclude that $\dim(\tilde{N}) = \dim \; \mathcal{A} S V^*$. The right-action of $S$ stretches $\mathcal{A}$ and embeds it into $\AR^{D \times N}$. The action of $V^*$ then simply turns the rows of the resulting matrices according to $V$. Consequently,
\[
	\dim(\tilde{N}) = \dim(\mathcal{A} S V^*) = \dim(\mathcal{A}) = \frac{1}{2} D (D-1).
\]

\end{proof}

\begin{lemma}\label{Lee.proposition.21.7}
	Let $\cdot: \; G \times M \rightarrow M$, $(g,P) \mapsto g \cdot P$ be a smooth left action of a compact Lie group $G$ on an smooth manifold $M$,  $P \in M$, and
	\[
		G_{P} = \{ g \in G \; | \; g \cdot P = P \}
	\]
	be the stabilizer of $P$. Suppose the stabilizer of $P$ is equal to the identity, i.e., $G_{P} = \{ e \}$. Then, the orbit $G \cdot P \subseteq M$ is a smooth submanifold of $M$.
\end{lemma}

\begin{proof}
	This is simply a rewriting of Corollary~21.6 and Proposition~21.7 in~\cite{Lee2013}.
\end{proof}

Recall that a configuration $P \in \AR^{D \times N}$ is \emph{generic} if it is sampled from a probability measure satisfying~\eqref{Rigidity:the.abs.cont.assumption}. We represent our knowledge $G_{\Omega}(P)$ about the unknown states-measurements Gram matrix by a vector $\knowledge \in \AR^{| \Omega |}$, i.e., $G_{\Omega}(P) = \knowledge$. An open neighborhood $U_{P}$ of a matrix $P \in \AR^{D \times N}$ is an open set containing $P$. In the remainder of this section we are going to show that---for generic configurations $P$---there exists an open and full-measure neighborhood $U_{P} \subseteq \AR^{D \times N}$ of the true configuration $P$ such that $U_{P} \cap G_{\Omega}^{-1}(\knowledge)$ is a submanifold of $\AR^{D \times N}$ (see Theorem~\ref{main.theorem.about.local.rigidity}).

For this purpose we will apply the constant-rank level set Theorem. 

\begin{theorem}[Theorem~5.12 in~\cite{Lee2013}]\label{Constant.rank.level.set.Thm}
	Let $M$ and $N$ be smooth manifolds, and let $\Phi: M \rightarrow N$ be a smooth map with constant rank $r$. Each level set of $\Phi$ is a properly embedded submanifold of codimension $r$ in $M$.
\end{theorem}

We are going to apply Theorem~\ref{Constant.rank.level.set.Thm} in the form
\beq\begin{aligned}
	M	&\mapsto	U \text{ open ($\Rightarrow$ a submanifold)},\\
	N 	&\mapsto	\AR^{| \Omega |},\\
	\Phi	&\mapsto	G_{\Omega}.
\end{aligned}\eeq
Hence, we need to show that there exists an open neighborhood $U \subseteq \AR^{D \times N}$ of the true configuration $P$ such that $\rank(d(G_{\Omega})_{(\cdot)})$ is constant on all of $U$. For this purpose we will require the following result about the decomposition of semi-algebraic sets.

\begin{theorem}[Proposition~2.9.10 in~\cite{Bochnak1998}]\label{Thm.decomposition.into.Nash.mfds}
	Let $S \subset \AR^n$ be a semi-algebraic set. Then, $S$ is the disjoint union of (Nash) submanifolds $M_{i}$, each (Nash) diffeomorphic to the open cube $(0,1)^{\times d_{i}}$ ($d_{i} \leq n$).
\end{theorem}


\subsubsection{Existence of the open neighborhood $U_{P}$}

For $\kappa \in \mathbb{N}$, let  
\[
	\mathcal{R}_{\kappa} := \{ Q \in \AR^{D \times N} \; | \; \rank \left(  d(G_{\Omega})(Q)  \right) = \kappa  \}.
\]
The rank of a matrix is the size of the largest non-zero minor. This motivates the definition of the polynomial
\[
	\omega_{\kappa}(Q) := \sum_{\substack{A \subseteq  d(G_{\Omega})_{Q} \\ A \in \mathbb{R}^{\kappa \times \kappa}  }} \left( \det A \right)^2
\]
($A \subseteq d(G_{\Omega})_{Q}$ means that $A$ is a submatrix of $d(G_{\Omega})_{Q}$), and its associated zero set
\[
	\mathcal{N}_{\kappa} := \{ Q \in \AR^{D \times N} \; | \; \omega_{\kappa}(Q) = 0 \}.
\]
Here, $d(G_{\Omega})_{Q}$ denotes the Jacobian of the map $G_{\Omega}$ evaluated at $Q$. The index set $\Omega$ is a subset of $DN$ matrix indices. Hence, $| \Omega | \leq ND$. The rank of $d(G_{\Omega})_{Q}$ is upper bounded by $| \Omega |$ because it is a $(| \Omega | \times ND)$-matrix. Moreover,
\beq
	\mathcal{R}_{\kappa} = \left\{ \begin{array}{ll}  \mathcal{N}_{\kappa+1} \cap \Bigl( \AR^{D \times N} \setminus \mathcal{N}_{\kappa} \Bigr) & \forall \kappa < | \Omega |,   \\    \AR^{D \times N} \setminus \mathcal{N}_{| \Omega |} & \forall \kappa = | \Omega | .  \end{array} \right. 
\eeq

\begin{lemma}\label{N.kappa.is.semialgebraic}
	$\mathcal{N}_{\kappa} \subseteq \AR^{D \times N}$ is an algebraic set.
\end{lemma}

\begin{proof}
	It is very easy to see that $\omega_{\kappa}(Q)$ is a polynomial in $( Q_{ij} )_{ij}$: All entries in $G_{\Omega}(Q)$ are inner products between different columns of configurations $Q \in \AR^{D \times N}$. The entries of the Jacobian $d(G_{\Omega})_{Q}$ are partial derivatives of the polynomial entries of $G_{\Omega}(Q)$. Therefore, $\bigl( d(G_{\Omega})_{Q} \bigr)_{ij}$ is polynomial in the entries of $Q$. The determinant of a submatrix $A \subseteq d(G_{\Omega})_{Q}$ is a polynomial in the entries of $A$ which are polynomials themselves. Thus, $\det(A)$ is a polynomial for all $A \subseteq d(G_{\Omega})_{Q}$. It follows that $\omega_{\kappa}(Q)$ is a polynomial in the entries of $Q$. We conclude that $\mathcal{N}_{\kappa} = \{ \omega_{\kappa}(Q) = 0 \}$ is an algebraic set.
\end{proof}

\begin{lemma}\label{N.closed.zero-set.lemma}
	Either $\mathcal{N}_{\kappa}$ is a closed set having measure zero, or $\mathcal{N}_{\kappa} = \AR^{D \times N}$.
\end{lemma}

\begin{proof}

The set $\mathcal{N}_{\kappa}$ is the preimage of the closed set $\{ 0 \}$, and $\omega_{\kappa}$ is a continuous function. Therefore, $\mathcal{N}_{\kappa}$ is closed. Moreover, $\mathcal{N}_{\kappa}$ is semialgebraic (see Lemma~\ref{N.kappa.is.semialgebraic}). Therefore (see Theorem~\ref{Thm.decomposition.into.Nash.mfds}), $\mathcal{N}_{\kappa}$ is the disjoint union of finitely many (Nash) submanifolds $M_{i}$, 
\[
	\mathcal{N}_{\kappa} = \bigcup_{i} M_{i},
\]
and each $M_{i}$ is (Nash) diffeomorphic to an open cube $(0,1)^{\times \dim(M_{i})}$. Let $\varphi_{i}: \; M_{i} \rightarrow (0,1)^{\times \dim(M_{i})}$ denote the corresponding diffeomorphisms. In the remainder we are going to distinguish the two cases $\max_{i} \bigl( \dim(M_{i}) \bigr) = ND$, and $\max_{i} \bigl( \dim(M_{i}) \bigr) < ND$.

Assume $\max_{i} \bigl( \dim(M_{i}) \bigr) = ND$. Consequently, there exists an index $j$ such that $M_{j}$ is a smooth submanifold in $\AR^{D \times N}$ of codimension 0. Hence, it is an open set (see for instance Proposition~5.1 in~\cite{Lee2013}). Consequently, $\exists \varepsilon > 0$ and $E \in \mathcal{N}_{\kappa}$, such that $B_{\varepsilon}(E) \subseteq \mathcal{N}_{\kappa}$. Expanding the polynomial $\omega_{\kappa}(Q)$ around $E$, we arrive at
\beq\begin{aligned}
	\omega_{\kappa}(Q) 
	&=	\sum_{i_{11}, i_{21},....,i_{DN}=0}^{\deg(\omega_{\kappa})} a_{i_{11},i_{21},\cdots, i_{DN}} Q_{11}^{i_{11}} Q_{21}^{i_{21}} \cdots Q_{DN}^{i_{DN}}\\
	&=	\sum_{i_{11}, i_{21},....,i_{DN}=0}^{\deg(\omega_{\kappa})} \tilde{a}_{i_{11},i_{21},\cdots, i_{DN}} (Q_{11}-E_{11})^{i_{11}}  \cdots (Q_{DN}-E_{DN})^{i_{DN}}.\nn
\end{aligned}\eeq
Setting $Q=E$, we find
\[
	0 = \omega_{\kappa}(E) =   \tilde{a}_{0, \cdots, 0}
\]
because by construction, $\omega_{\kappa}$ vanishes on all of $B_{\varepsilon}(E)$. Analogously ($\omega_{\kappa}$ is constant on $B_{\varepsilon}(E)$),
\[
	0 = \left. \frac{\partial}{\partial Q_{11}^{i_{11}} \cdots \partial Q_{DN}^{i_{DN} }} \omega_{\kappa}(Q) \right|_{Q=E} = \tilde{a}_{i_{11},\cdots,i_{DN}}
\]
for all $(i_{11}, ..., i_{DN})$. Hence, $\omega_{\kappa}(Q)$ is equal to zero on all of $\AR^{D \times N}$, and we conclude $\mathcal{N}_{\kappa} = \AR^{D \times N}$.

Assume $\max_{i} \bigl( \dim(M_{i}) \bigr) < ND$. The map $\varphi_{i}^{-1}$ restricted to $(0,1)^{\times \dim(M_{i})}$ is a smooth map from the $\dim(M_{i})$-dimensional manifold $(0,1)^{\times \dim(M_{i})}$ to the linear manifold $\AR^{D \times N}$. By assumption $\dim(M_{i}) < DN$. Therefore, $\rank( d(\varphi_{i}^{-1})_{Z} ) \leq \min\{ \dim(M_{i}),DN \} < DN$ for all $Z \in (0,1)^{\times \dim(M_{i})}$, i.e., each element of $M_{i}$ is a critical value. The set of critical values is a set of measure zero (Sard's Theorem). We conclude that $\mathcal{N}_{\kappa}$ is a set of measure zero because it is the finite union of the measure zero sets $M_{i}$.

\end{proof}

\begin{lemma}\label{inclusions.of.N.kappa.lemma}
	$\mathcal{N}_{\kappa_{2}} \subseteq \mathcal{N}_{\kappa_{1}}$ for all $\kappa_{2} < \kappa_{1}$.
\end{lemma}

\begin{proof}
	Let $\kappa_{2} < \kappa_{1}$ and $M \in \mathcal{N}_{\kappa_{2}} =  \{ \omega_{\kappa_{2}} = 0 \}$. Hence, all $(\kappa_{2} \times \kappa_{2})$-minors of $d(G_{\Omega})_{M}$ are equal to zero. It follows that all $(\kappa_{1} \times \kappa_{1})$-minors of $d(G_{\Omega})_{M}$ are equal to zero because they can be written as weighted sum of $(\kappa_{2} \times \kappa_{2})$-minors (Laplace expansion; recall $\kappa_{2} < \kappa_{1}$). Therefore, $\omega_{\kappa_{1}}(M) = 0$ so that $M \in \mathcal{N}_{\kappa_{1}}$.
\end{proof}

\begin{lemma}\label{Lemma.about.existence.of.open.constant.rk.NH}
	Let $P \in \AR^{D \times N}$ be a generic states-measurements configuration. Then, with probability 1, there exists an open neighborhood $U_{P} \subseteq \AR^{D \times N}$ of $P$ and $r \in \{ 1,..., | \Omega | \}$ such that $\rank(d(G_{\Omega})_{Q}) = r$ for all $Q \in U_{P}$. Moreover, $U_{P}$ is full measure.
\end{lemma}

\begin{proof}

Lemma~\ref{inclusions.of.N.kappa.lemma} implies
\beq\label{rank.sets.as.Nk+1.minus.Nk}
	\mathcal{R}_{\kappa} = \left\{ \begin{array}{ll}  \mathcal{N}_{\kappa+1}  \setminus \mathcal{N}_{\kappa}  & \forall \kappa < | \Omega |,   \\    \AR^{D \times N} \setminus \mathcal{N}_{| \Omega |} & \forall \kappa = | \Omega | .  \end{array} \right. 
\eeq
According to Lemma~\ref{N.closed.zero-set.lemma}, $\mathcal{N}_{\kappa}$ is either a closed set of measure zero, or $\mathcal{N}_{\kappa} = \AR^{D \times N}$. 

Assume that there exists no $\kappa \in \{ 1,..., | \Omega | \}$, such that $\mathcal{N}_{\kappa}  = \AR^{D \times N}$. 
Then, all the sets $\mathcal{N}_{\kappa}$ are closed sets of measure zero. Hence, by Eq.~\eqref{rank.sets.as.Nk+1.minus.Nk},
\beq\label{rank.sets.in.terms.of.zero.sets}
	\mathcal{R}_{\kappa} = \left\{ \begin{array}{ll}  
	(\text{set of measure zero})  & \forall \kappa < | \Omega |,   \\    
	(\text{open set of full measure}) & \forall \kappa = | \Omega |. 
	\end{array} \right. 
\eeq
It follows that $\mathbb{P}[ P \in \mathcal{R}_{| \Omega |} ] = 1$ because $P$ is generic. We conclude that with probability 1, $U_{P} := \mathcal{R}_{| \Omega |}$ is an open neighborhood of $P$ with the desired properties. This concludes the proof for this particular scenario.

This leaves us with the opposite scenario: assume that there exists $k \in \{ 1,..., | \Omega | \}$, such that $\mathcal{N}_{k}  = \AR^{D \times N}$. Set 
\[
	\tilde{k} := \min\{ k | \mathcal{N}_{k} =  \AR^{D \times N}\}.
\]
Note that $\tilde{k} > 1$, because otherwise, $d(G_{\Omega})_{Q} = 0$ for all $Q \in \AR^{D \times N}$ (this is impossible because every $(i,j) \in \Omega$ leads to a non-vanishing row in the Jacobian $d(G_{\Omega})_{Q}$). Lemma~\ref{inclusions.of.N.kappa.lemma} implies
\beq
	\mathcal{N}_{\kappa} = \left\{ \begin{array}{ll}  
	(\text{closed set of measure zero})  & \forall \kappa < \tilde{k},   \\    
	\AR^{D \times N} & \forall \kappa \geq \tilde{k}. 
	\end{array} \right. 
\eeq	
Hence (see Eq.~\eqref{rank.sets.as.Nk+1.minus.Nk} and \eqref{rank.sets.in.terms.of.zero.sets}),
\beq\begin{aligned}
	\mathcal{R}_{1}
	&=	(\text{set of measure zero}) \setminus (\text{set of measure zero}),\\
	\vdots \
	&=	\ \vdots\\
	\mathcal{R}_{\tilde{k}-2}
	&=	(\text{set of measure zero}) \setminus (\text{set of measure zero}),\\
	\mathcal{R}_{\tilde{k}-1}
	&=	\AR^{D \times N} \setminus (\text{closed set of measure zero}),\\
	\mathcal{R}_{\tilde{k}}
	&=	\AR^{D \times N} \setminus \AR^{D \times N} = \emptyset, \\
	\vdots \
	&=	\ \vdots\\
	\mathcal{R}_{| \Omega |}
	&=	\AR^{D \times N} \setminus \AR^{D \times N} = \emptyset.
\end{aligned}\eeq
We observe that $\mathcal{R}_{\tilde{k}-1}$ is an open set of full measure. Hence, $\mathbb{P}[ P \in \mathcal{R}_{\tilde{k}-1} ] = 1$ because $P$ is generic. We conclude that with probability 1, $\mathcal{R}_{\tilde{k}-1}$ is an open neighborhood of $P$ with the desired properties. This concludes the proof of the Lemma.
\end{proof}

\begin{theorem}\label{main.theorem.about.local.rigidity}
	Let $P \in \AR^{D \times N}$ be generic and $\vec{\mathcal{K}} = G_{\Omega}(P)$. Then (with probability 1) there exists an open and full measure neighborhood $U_{P}$ of $P$ such that $U_{P} \cap G_{\Omega}^{-1}(\vec{\mathcal{K}})$ is a smooth $(DN - r)$-dimensional submanifold of $\AR^{D \times N}$.
\end{theorem}

\begin{proof}
	According to Lemma~\ref{Lemma.about.existence.of.open.constant.rk.NH} there exists an open, full measure neighborhood $U_{P}$ of $P$ and $r \in \{ 1,..., | \Omega | \}$ such that $\rank(d(G_{\Omega})_{Q}) = r$ for all $Q \in U_{P}$. Therefore, $U_{P}$ is a submanifold of $\AR^{D \times N}$. Consequently, $G_{\Omega}$ is a smooth map with constant rank $r$ which maps the smooth manifold $U_{P}$ into the smooth manifold $\AR^{| \Omega |}$. This shows that  the constant rank level set Theorem (see Theorem~\ref{Constant.rank.level.set.Thm}) is applicable.
\end{proof}



%

\end{document}